\documentclass{journal}
\pdfoutput=1

\usepackage{amsmath,amsthm,amssymb}
\usepackage{graphicx,epstopdf}
\usepackage{url,hyperref}
\usepackage{color}
\usepackage{microtype}
\usepackage{plaatjes}
\usepackage{fullpage}
\usepackage[left,pagewise]{lineno}
\usepackage{xspace}
\usepackage{euscript}
\usepackage[section]{algorithm}
\usepackage[noend]{algorithmic}

\definecolor {infocolor} {rgb} {0.6,0.6,0.6}

\definecolor {remarkcolor} {rgb} {0.1,0.7,0.2}
\definecolor {namecolor} {rgb} {0.2,0.1,0.7}

\newcommand{\eps}{\varepsilon}

\newcommand{\Eu}[1]{\ensuremath{\EuScript{#1}}}

\renewcommand{\b}[1]{\ensuremath{\mathbb{#1}}}

\newcommand{\R}{\ensuremath{\mathbb{R}}}

\newcommand{\diameter}{\ensuremath{\textsf{diam}}}
\newcommand{\seb}{\ensuremath{\textsf{seb}_2}}
\newcommand{\dwid}{\ensuremath{\textsf{dwid}}}

\newcommand{\sip}{\textsf{sip\xspace}}
\newcommand{\IP}[2]{\ensuremath{ \langle #1 , #2 \rangle}}

\newcommand{\wid}{\omega}
\newcommand{\etal}{et. al.}
\newcommand{\pset}{support\xspace}
\newcommand{\psets}{supports\xspace}

\newcommand{\PDIAM}{$\ensuremath{\textrm{PLANAR-DIAM}}$\xspace}

\newcommand{\twoSAT}{$\ensuremath{\textrm{\#2SAT}}$\xspace}

\newtheorem {theorem}{Theorem}[section]
\newtheorem {problem}[theorem]{Problem}
\newtheorem {lemma}[theorem]{Lemma}

\newtheorem {corollary}[theorem]{Corollary}

\title{Geometric Computations\\ on Indecisive and Uncertain Points 
}

\author{
  Allan J{\o}rgensen\thanks
  { formerly: MADALGO, Deptartment of Computer Science, University of Aarhus, Denmark.
    \texttt{jallan(at)madalgo.au.dk}
  }
  \and
  Maarten L\"offler\thanks
  {
    formerly: Computer Science Deptartment, University of California, Irvine, USA.
    \texttt{mloffler(at)uci.edu}
  }
  \and
  Jeff M. Phillips\thanks
  {
    School of Computing, University of Utah, USA.
    \texttt{jeffp(at)cs.utah.edu}
  }
}
\date{ }

\begin{document}

\maketitle

\begin{abstract}
We study computing geometric problems on \emph {uncertain} points. An uncertain point is a point that does not have a fixed location, but rather is described by a probability distribution.  When these probability distributions are restricted to a finite number of locations, the points are called \emph{indecisive} points.  In particular, we focus on geometric shape-fitting problems and on building compact distributions to describe how the solutions to these problems vary with respect to the uncertainty in the points.  Our main results are:
(1) a simple and efficient randomized approximation algorithm for calculating the distribution of any statistic on uncertain data sets;
(2) a polynomial, deterministic and exact algorithm for computing the distribution of answers for any LP-type problem on an indecisive point set; and
(3) the development of \emph{shape inclusion probability} (SIP) functions which captures the ambient distribution of shapes fit to uncertain or indecisive point sets and are admissible to the two algorithmic constructions.  
\end{abstract}

\section {Introduction}

In gathering data there is a trade-off between quantity and accuracy.  The drop in the price of hard drives and other storage costs has shifted this balance towards gathering enormous quantities of data, yet with noticeable and sometimes intentionally tolerated increased error rates.  However, often as a benefit from the large data sets, models are developed to describe the pattern of the data error.

Let us take as an example Light Detection and Ranging (LIDAR) data gathered for Geographic Information Systems (GIS)~\cite{LKC04}, specifically height values at millions of locations on a terrain.
Each data point $(x,y,z)$ has an $x$-value (longitude), a $y$-value (latitude), and a $z$-value (height).
This data set is gathered by a small plane flying over a terrain with a laser aimed at the ground measuring the distance from the plane to the ground.  Error can occur due to inaccurate estimation of the plane's altitude and position or artifacts on the ground distorting the laser's distance reading.
But these errors are well-studied and can be modeled by replacing each data point with a probability distribution of its actual position.  Greatly simplifying, we could represent each data point as a $3$-variate normal distribution centered at its recorded value; in practice, more detailed uncertainty models are built.

Similarly, large data sets are gathered and maintained for many other applications.
In robotic mapping~\cite{Thr02,EP04} error models are provided for data points gathered by laser range finders and other sources.
In data mining~\cite{AY08,AS00} original data (such as published medical data) are often perturbed by a known model to preserve anonymity.
In spatial databases~\cite{GS05,SC01,CKP03} large data sets may be summarized as probability distributions to store them more compactly.
Data sets gathered by crawling the web have many false positives, and allow for models of these rates.  
Sensor networks~\cite{DGMHH04} stream in large data sets collected by cheap and thus inaccurate sensors.
In protein structure determination~\cite{PYCDB06} every atom's position is imprecise due to inaccuracies in reconstruction techniques and the inherent flexibility in the protein.
In summary, there are many large data sets with modeled errors and this uncertainty should be dealt with explicitly.

\subsection {The Input: Geometric Error Models}

The input for a typical computational geometry problem is a set $P$ of $n$ points in $\R^2$, or more generally $\R^d$.  In this paper we consider extensions of this model where each point is also given a model of its uncertainty.  This model describes for each point a distribution or bounds on the point's location, if it exists at all.  

\begin{itemize}
\item 
Most generally, we describe these data as \emph{uncertain points} $\Eu{P} = \{P_1, P_2, \ldots, P_n\}$.  Here each point's location is described by a probability distribution $\mu_i$ (for instance by a Gaussian distribution).  This general model can be seen to encompass the forthcoming models, but is often not worked with directly because of the computational difficulties arisen from its generality.  For instance in tracking uncertain objects a particle filter uses a discrete set of locations to model uncertainty~\cite{MDFW00} while a Kalman filter restricts the uncertainty model to a Gaussian distribution~\cite{Kal60}.  

\item 
A more restrictive model we also study in this paper are \emph{indecisive points} where each point can take one of a finite number of locations.  To simplify the model (purely for making results easier to state) we let each point have exactly $k$ possible locations, forming the domain of a probability distribution.  That is each uncertain point $P_i$ is at one of $\{p_{i,1}, p_{i,2}, \ldots, p_{i,k}\}$.  Unless further specified, each location is equally likely with probability $1/k$, but we can also assign each location a weight $w_{i,j}$ as the probability that $P_i$ is at $p_{i,j}$ where $\sum_{j=1}^k w_{i,j} = 1$ for all $i$.  

\tweeplaatjes {ex-input} {ex-sample} {(a) An example input consisting of $n = 3$ sets of $k = 6$ points each. (b) One of the $6^3$ possible samples of $n = 3$ points.}

Indecisive points appear naturally in many applications. They play an important role in databases~\cite{DS04,ABSHNSW06,CM08,CG09,TCXNKP05,ACTY09,CLY09}, machine learning~\cite{BZ04}, and sensor networks~\cite{ZC04} where a limited number of probes from a certain data set are gathered, each potentially representing the true location of a data point.  Alternatively, data points may be obtained using imprecise measurements or are the result of inexact earlier computations.
However, the results with detailed algorithmic analysis generally focus on one-dimensional data; furthermore, they often only return the expected value or the most likely answer instead of calculating a full distribution.

\item 
An \emph{imprecise point} is one where its location is not known precisely, but it is restricted to a range.  In one-dimension these ranges are modeled as uncertainty intervals, but in 2 or higher dimensions they become geometric regions. 
An early model to quantify imprecision in geometric data, motivated by finite precision of coordinates, is {\it $\eps$-geometry}, introduced by Guibas \etal~\cite {gss-egbra-89}, where each point was only known to be somewhere within an $\eps$-radius ball of its guessed location. 
The simplicity of this model has provided many uses in geometry.
Guibas \etal~\cite {gss-cscah-93} define \emph {strongly convex} polygons: polygons that are guaranteed to stay convex, even when the vertices are perturbed by $\eps$. 
Bandyopadhyay and Snoeyink~\cite{bs-ads-04} compute the set of all potential simplices in $\b R^2$ and $\b R^3$ that could belong to the Delaunay triangulation.
Held and Mitchell~\cite {hm-ticpps-08} and L\"offler and Snoeyink~\cite {ls-dtip-08} study the problem of preprocessing a set of imprecise points under this model, so that when the true points are specified later some computation can be done faster.

A more involved model for imprecision can be obtained by not specifying a single $\eps$ for all the points, but allowing a different radius for each point, or even other shapes of imprecision regions. This allows for modeling imprecision that comes from different sources, independent imprecision in different dimensions of the input, etc. This extra freedom in modeling comes at the price of more involved algorithmic solutions, but still many results are available.
Nagai and Tokura~\cite {nt-teb-00} compute the union and intersection of all possible convex hulls to obtain bounds on any possible solution, as does Ostrovsky-Berman and Joskowicz~\cite {obj-ue-05} in a setting allowing some dependence between points.
Van Kreveld and L\"offler~\cite {kl-bgmips-06} study the problem of computing the smallest and largest possible values of several geometric extent measures, such as the diameter or the radius of the smallest enclosing ball, where the points are restricted to lie in given regions in the plane.
Kruger~\cite {k-bmips-08} extends some of these results to higher dimensions.

Although imprecise points do not traditionally have an associated probability distribution associated to them, we argue that they can still be considered a special case of our uncertain points, since we can impose e.g. a uniform distribution on the regions, and then ask question about the smallest or largest non-zero probability values of some function, which would correspond to bounds in the classical model.

\item 
A \emph{stochastic point} $p$ has a fixed location, but which only exists with a probability $\rho$.  These points arise naturally in many database scenarios~\cite{ABSHNSW06,CLY09} where gathered data has many false positives. Recently in geometry Kamousi, Chan, and Suri~\cite{KCS11a,KCS11b} considered geometric problems on stochastic points and geometric graphs with stochastic edges.  
These stochastic data sets can be interpreted as uncertain point sets as well by allowing the probability distribution governing uncertain points to have a certain probability of not existing, or rather the integral of the distribution is $\rho$ instead of always $1$.  

\end{itemize}

\subsection {The Output: Distributional Representations}

This paper studies how to compute distributions of statistics over uncertain data.  
These distributions can take several forms.  In the simplest case, a distribution of a single value has a one-dimensional domain.  
The technical definition yields a simpler exposition when the distribution is represented as a cumulative density function, which we refer to as a \emph{quantization}.  
This notion can be extended to a multi-dimensional cumulative density function (a \emph{$k$-variate quantization}) as we measure multiple variables simultaneously.  Finally, we also describe distributions over shapes defined on uncertain points (e.g. minimum enclosing ball).  As the domains of these shape distributions are a bit abstract and difficult to work with, we convey this information as a \emph{shape inclusion probability} or \emph{SIP}; for any point in the domain of the input point sets we describe the probability the point is contained in the shape.  

This model of uncertain data has been studied in the database community but for different types of problems on usually one-dimensional data, such as indexing~\cite{ACTY09,TCXNKP05,KMMH06}, ranking~\cite{CLY09}, nearest neighbors~\cite{CCMC08} and creating histograms~\cite{CG09}.

\subsection{Contributions}

For each type of distributional representation of output we study, the goal is a function from some domain to a range of $[0,1]$.  
For the general case of uncertain points, we provide simple and efficient, randomized approximation algorithms that results in a function that \emph{everywhere} has error at most $\eps$.  Each variation of the algorithm runs in $O((1/\eps^2)(\nu + \log (1/\delta)) T)$ time and produces an output of size $O((1/\eps^2)(\nu + \log (1/\delta))$ where $\nu$ describes the complexity of the output shape (i.e. VC-dimension), $\delta$ is the probability of failure, and $T$ is the time it takes to compute the geometric question on certain points.  These results are quite practical as experimental results demonstrate that the constant for the big-Oh notation is approximately $0.5$.  Furthermore, for one-dimensional output distributions (quantizations) the size can be reduced to $1/\eps$, and for $k$-dimensional distributions to $O((k/\eps) \log^4 (1/\eps))$.  
We also extend these approaches to allow for geometric approximations based on $\alpha$-kernels~\cite{AHV04,AHV07}.  

For the case of indecisive points, we provide deterministic and exact, polynomial-time algorithms for all LP-type problems with constant combinatorial dimension (e.g. minimum enclosing ball).  We also provide evidence that for problems outside this domain, deterministic exact algorithms are not available, in particular showing that diameter is \#P-hard despite having an output distribution of polynomial size.  
Finally, we consider deterministic algorithms for uncertain point sets with continuous distributions describing the location of each point.  We describe a non-trivial range space on these input distributions from which an $\eps$-sample creates a set of indecisive points, from which this algorithm can be performed to deterministically create an approximation to the output distribution.

\section {Preliminaries}\label{sec:prelim}

This section provides formal definitions for existing approximation schemes related to our work as well as the for our output distributions.  

\subsection{Approximation Schemes: $\eps$-Samples and $\alpha$-Kernels}

This work allows for three types of approximations.  The most natural in this setting is controlled by a parameter $\eps$ which denotes the error tolerance for probability.  That is an $\eps$-approximation for any function with range in $[0,1]$ measuring probability can return a value off by at most an additive $\eps$.  
The second type of error is a parameter $\delta$ which denotes the chance of failure of a randomized algorithm.  That is a $\delta$-approximate randomized algorithm will be correct with probability at least $1-\delta$.  
Finally, in specific contexts we allow a geometric error parameter $\alpha$.  In our context, an $\alpha$-approximate geometric algorithm can allow a relative $\alpha$-error in the width of any object.  This is explained more formally below.  
It should be noted that these three types of error cannot be combined into a single term, and each needs to be considered separately.  However, the $\eps$ and $\delta$ parameters have a well-defined trade-off.  

\paragraph{$\eps$-Samples.}
For a set $P$ let $\Eu{A}$ be a set of subsets of $P$.  In our context usually $P$ will be a point set and the subsets in $\Eu{A}$ could be induced by containment in a shape from some family of geometric shapes.
For example, 
 $\Eu{I}_+$ describes one-sided intervals of the form $(-\infty, t)$.
The pair $(P, \Eu{A})$ is called a \emph{range space}.  We say that $Q \subset P$ is an \emph{$\eps$-sample} of $(P,\Eu{A})$ if
\[
\forall_{R \in \Eu{A}} \left|\frac{\phi(R \cap Q)}{\phi(Q)} - \frac{\phi(R \cap P)}{\phi(P)}\right| \leq \eps,
\]
where $|\cdot|$ takes the absolute value and $\phi(\cdot)$ returns the measure of a point set.  In the discrete case $\phi(Q)$ returns the cardinality of $Q$.  We say $\Eu{A}$ \emph{shatters} a set $S$ if every subset of $S$ is equal to $R \cap S$ for some $R \in \Eu{A}$.  The cardinality of the largest discrete set $S \subseteq P$ that $\Eu{A}$ can shatter is the \emph{VC-dimension} of $(P,\Eu{A})$.

When $(P, \Eu{A})$ has constant VC-dimension $\nu$, we can create an $\eps$-sample $Q$ of $(P, \Eu{A})$, with probability $1-\delta$, by uniformly sampling $O((1/\eps^2)(\nu + \log (1/\delta) ))$ points from $P$~\cite{VC71,LLS01}.  There exist deterministic techniques to create $\eps$-samples~\cite{Mat91,CM96} of size $O(\nu (1/\eps^2) \log (1/\eps))$ in time $O(\nu^{3\nu} n ((1/\eps^2) \log (\nu/\eps))^\nu)$.
A recent result of Bansal~\cite{Ban10} (also see this simplification~\cite{LM12}) can slightly improve this bound to $O(1/\eps^{2-1/2v})$, following an older existence proof~\cite{MWW93}, in time polynomial in $n$ and $1/\eps$.  
When $P$ is a point set in $\R^d$ and the family of ranges $\Eu{R}_d$ is determined by inclusion in axis-aligned boxes, then an $\eps$-sample for $(P, \Eu{R}_d)$ of size $O((d/\eps) \log^{2d} (1/\eps))$ can be constructed in $O((n/\eps^3) \log^{6d} (1/\eps))$ time~\cite{Phi08}.

For a range space $(P, \Eu{A})$ the \emph{dual range space} is defined $(\Eu{A}, P^*)$ where $P^*$ is all subsets $\Eu{A}_p \subseteq \Eu{A}$ defined for an element $p \in P$ such that $\Eu{A}_p = \{A \in \Eu{A} \mid p \in A\}$.  If $(P,\Eu{A})$ has VC-dimension $\nu$, then $(\Eu{A}, P^*)$ has VC-dimension $\leq 2^{\nu+1}$.  Thus, if the VC-dimension of $(\Eu{A}, P^*)$ is constant, then the VC-dimension of $(P,\Eu{A})$ is also constant \cite{Mat99}.

When we have a distribution $\mu : \R^d \to \R^+$, such that $\int_{x \in \R} \mu(x) \; dx = 1$, we can think of this as the set $P$ of all points in $\R^d$, where the weight $w$ of a point $p \in \R^d$ is $\mu(p)$.  
To simplify notation, we write $(\mu, \Eu{A})$ as a range space where the ground set is this set $P = \R^d$ weighted by the distribution $\mu$.

\paragraph{$\alpha$-Kernels.}
Given a point set $P \in \R^d$ of size $n$ and a direction $u \in \b{S}^{d-1}$, let $P[u] = \arg\max_{p \in P} \IP{p}{u}$, where $\IP{\cdot}{\cdot}$ is the inner product operator.
Let $\wid(P,u) = \IP{P[u] - P[-u]}{u}$ describe the width of $P$ in direction $u$.
We say that $K \subseteq P$ is an \emph{$\alpha$-kernel} of $P$ if for all $u \in \b{S}^{d-1}$
\[
\wid(P,u) - \wid(K,u) \leq \alpha \cdot \wid(P,u).
\]
$\alpha$-kernels of size $O(1/\alpha^{(d-1)/2})$~\cite{AHV04} can be calculated in time $O(n + 1/\alpha^{d-3/2})$~\cite{Cha06,YAPV04}.  Computing many extent related problems such as diameter and smallest enclosing ball on $K$ approximates the problem on $P$~\cite{AHV04,AHV07,Cha06}.

\subsection{Problem Statement}  
Let $\mu_i : \R^d \to \R^+$ describe the probability distribution of an uncertain point $P_i$ where the integral $\int_{q \in \R^d} \mu_i(q) \; dq = 1$.
We say that a set $Q$ of $n$ points is a \emph {\pset} from $\Eu P$ if it contains exactly one point from each set $P_i$, that is, if $Q = \{q_1, q_2, \ldots, q_n\}$ with $q_i \in P_i$.
In this case we also write $Q \Subset \Eu P$.
Let $\mu_{\Eu{P}} : \R^d \times \R^d \times \ldots \times \R^d \to \R^+$ describe the distribution of supports $Q = \{q_1, q_2, \ldots, q_n\}$ under the joint probability over each $q_i \in P_i$.
For brevity we write the space $\R^d \times \ldots \times \R^d$ as $\R^{dn}$.
For this paper we will assume $\mu_{\Eu{P}}(q_1, q_2, \ldots, q_n) = \prod_{i=1}^n \mu_{p_i}(q_i)$, so the distribution for each point is independent, although  this restriction can be easily removed for all randomized algorithms.

\paragraph{Quantizations and their approximations.}
Let $f : \R^{dn} \to \R^k$ be a function on a fixed point set.  Examples include the radius of the minimum enclosing ball where $k=1$ and the width of the minimum enclosing axis-aligned rectangle along the $x$-axis and $y$-axis where $k=2$.
Define the ``dominates'' binary operator $\preceq$ so that $(p_1, \ldots, p_k) \preceq (v_1, \ldots, v_k)$ is true if for every coordinate $p_i \leq v_i$.
Let $\b{X}_f(v) = \{Q \in \R^{dn} \mid f(Q) \preceq v\}$.
For a query value $v$ define,
$
F_{\mu_P}(v) = \int_{Q \in \b{X}_f(v)} \mu_P(Q) \, dQ.
$
Then $F_{\mu_P}$ is the cumulative density function of the distribution of possible values that $f$ can take\footnote{For a function $f$ and a distribution of point sets $\mu_P$, we will always represent the cumulative density function of $f$ over $\mu_P$ by $F_{\mu_P}$.
}.  We call $F_{\mu_P}$ a \emph{quantization} of $f$ over $\mu_P$.  

\drieplaatjes {uq-true} {uq-points} {uq-approx}
{ \label{fig:1e-quants}
  (a) The true form of a monotonically increasing function from $\R \to \R$.
  (b) The $\eps$-quantization $R$ as a point set in $\R$.  
  (c) The inferred curve $h_R$ in $\R^2$.
}

Ideally, we would return the function $F_{\mu_P}$ so we could quickly answer any query exactly, however, for the most general case we consider, it is not clear how to calculate $F_{\mu_P}(v)$ exactly for even a single query value $v$.
Rather, we introduce a data structure, which we call an $\eps$-quantization, to answer any such query approximately and efficiently, illustrated in Figure \ref{fig:1e-quants} for $k=1$.  An \emph{$\eps$-quantization} is a point set $R \subset \R^k$ which induces a function $h_R$ where $h_R(v)$ describes the fraction of points in $R$ that $v$ dominates.  Let $R_v = \{r \in R \mid r \preceq v\}$.  Then $h_R(v) = |R_v|/|R|$.
For an isotonic (monotonically increasing in each coordinate) function $F_{\mu_P}$ and any value $v$, an $\eps$-quantization, $R$, guarantees that
$
|h_R(v) - F_{\mu_P}(v)| \leq \eps.
$
More generally (and, for brevity, usually only when $k>1$), we say $R$ is a $k$-variate $\eps$-quantization.  An example of a $2$-variate $\eps$-quantization is shown in Figure \ref{fig:ke-quants}.
The space required to store the data structure for $R$ is dependent only on $\eps$ and $k$, not on $|P|$ or $\mu_P$.

\vierplaatjes [height=100pt] {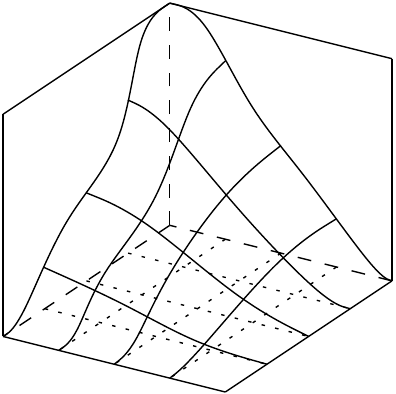} {mq-points} {mq-approx} {mq-both}
{ \label{fig:ke-quants}
  (a) The true form of a $xy$-monotone $2$-variate function.
  (b) The $\eps$-quantization $R$ as a point set in $\R^2$.
  (c) The inferred surface $h_R$ in $\R^3$.
  (d) Overlay of the two images.
}

\paragraph{$(\eps, \delta, \alpha)$-Kernels.}
Rather than compute a new data structure for each measure we are interested in, we can also compute a single data structure (a coreset) that allows us to answer many types of questions.
For an isotonic function $F_{\mu_P} : \R^+ \to [0,1]$, an \emph{$(\eps, \alpha)$-quantization} data structure $M$ describes a function $h_M : \R^+ \to [0,1]$ so for any $x \in \R^+$, there is an $x^\prime \in \R^+$ such that
(1) $|x-x^\prime| \leq \alpha x$ and
(2) $|h_M(x) - F_{\mu_P}(x^\prime)| \leq \eps$.
An \emph{$(\eps, \delta, \alpha)$-kernel} is a data structure that can produce an $(\eps, \alpha)$-quantization, with probability at least $1-\delta$, for $F_{\mu_P}$ where $f$ measures the width in any direction and whose size depends only on $\eps$, $\alpha$, and $\delta$.
The notion of $(\eps, \alpha)$-quantizations is generalizes to a $k$-variate version, as do $(\eps, \delta, \alpha)$-kernels. 

\paragraph{Shape inclusion probabilities.}
A \emph{summarizing shape} of a point set $P \subset \R^d$ is a Lebesgue-measureable subset of $\R^d$ that is determined by $P$.  I.e. given a class of shapes $\Eu{S}$, the summarizing shape $S(P) \in \Eu{S}$ is the shape that optimizes some aspect with respect to $P$.
Examples include the smallest enclosing ball and the minimum-volume axis-aligned bounding box.
For a family $\Eu{S}$ we can study the \emph{shape inclusion probability function} $s_{\mu_P} : \R^d \to [0,1]$ (or \sip\ function), where $s_{\mu_P}(q)$ describes the probability that a query point $q \in \R^d$ is included in the summarizing shape\footnote{For technical reasons, if there are (degenerately) multiple optimal summarizing shapes, we say each is equally likely to be the summarizing shape of the point set.}.
For the more general types of uncertain points, there does not seem to be a closed form for many of these functions.
In these cases we can calculate an $\eps$-\sip function $\hat{s} : \R^d \to [0,1]$ such that
$
\forall_{q \in \R^d} \left| s_{\mu_P}(q) - \hat{s}(q) \right| \leq \eps.
$
The space required to store an $\eps$-\sip function depends only on $\eps$ and the complexity of the summarizing shape.

\section{Randomized Algorithm for $\eps$-Quantizations}
\label{sec:rand-eQ}

We develop several algorithms with the following basic structure (as outlined in Algorithm \ref{alg:rand-draw}):
(1) sample one point from each distribution to get a random point set;
(2) construct the summarizing shape of the random point set;
(3) repeat the first two steps $O((1/\eps^2)(\nu + \log(1/\delta)))$ times and calculate a summary data structure.
This algorithm only assumes that we can draw a random point from $\mu_p$ for each $p \in P$ in constant time; if the time depends on some other parameters, the time complexity of the algorithms can be easily adjusted.

\begin{algorithm}[h!!t]
\caption{Approximate $\mu_P$ w.r.t. a family of shapes $\Eu{S}$ or function $f_{\Eu{S}}$
\label{alg:rand-draw}}
\begin{algorithmic}[1]
\FOR {$i = 1$ \textbf{to} $m = O((1/\eps^2) (\nu + \log (1/\delta)))$}
  \FOR {\textbf{all} $p_j \in P$}
    \STATE Sample $q_j$ from $\mu_{p_j}$.
  \ENDFOR
  \STATE Set $V_i = f_{\Eu{S}}(\{q_1, q_2, \ldots, q_n\})$.
\ENDFOR
\STATE Reduce or Simplify the set $\Eu{V} = \{ V_i\}_{i=1}^m$.
\end{algorithmic}
\end{algorithm}

\paragraph{Algorithm for $\eps$-quantizations.}
\label{sec:algQ}
For a function $f$ on a point set $P$ of size $n$, it takes $T_f(n)$ time to evaluate $f(P)$.
We construct an approximation to $F_{\mu_P}$ 
as follows.  First draw a sample point $q_j$ from each $\mu_{p_j}$ for $p_j \in P$, then evaluate $V_i = f(\{q_1, \ldots, q_n\})$.  The fraction of trials of this process that produces a value dominated by $v$ is the estimate of $F_{\mu_P}(v)$.
In the univariate case we can reduce the size of $\Eu{V}$ by returning $2/\eps$ evenly spaced points according to the sorted order.

\begin{theorem}
For a distribution $\mu_P$ of $n$ points, with success probability at least $1-\delta$,
there exists an $\eps$-quantization of size $O(1/\eps)$ for $F_{\mu_P}$, and it can be constructed in $O(T_f(n) (1/\eps^2) \log (1/\delta))$ time.
\label{thm:eq-main}
\end{theorem}
\begin{proof}
Because $F_{\mu_P} : \R \to [0,1]$ is an isotonic function, there exists another function $g: \R \to \R^+$ such that $F_{\mu_P}(t) = \int_{x = -\infty}^t g(x) \; dx$ where $\int_{x \in \R} g(x) \;dx = 1$.  Thus $g$ is a probability distribution of the values of $f$ given inputs drawn from $\mu_P$.
This implies that an $\eps$-sample of $(g, \Eu{I}_+)$ is an $\eps$-quantization of $F_{\mu_P}$, since both estimate within $\eps$ the fraction of points in any range of the form $(-\infty, x)$.  This last fact can also be seen through a result by Dvoretzky, Kiefer, and Wolfowitz~\cite{DKW56}.  

By drawing a random sample $q_i$ from each $\mu_{p_i}$ for $p_i \in P$, we are drawing a random point set $Q$ from $\mu_P$.  Thus $f(Q)$ is a random sample from $g$.  Hence, using the standard randomized construction for $\eps$-samples, $O((1/\eps^2) \log (1/\delta))$ such samples will generate an $(\eps/2)$-sample for $g$, and hence an $(\eps/2)$-quantization for $F_{\mu_P}$, with probability at least $1-\delta$.

Since in an $(\eps/2)$-quantization $R$ every value $h_R(v)$ is different from $F_{\mu_P}(v)$ by at most $\eps/2$, then we can take an $(\eps/2)$-quantization of the function described by $h_R(\cdot)$ and still have an $\eps$-quantization of $F_{\mu_P}$.
Thus, we can reduce this to an $\eps$-quantization of size $O(1/\eps)$ by taking a subset of $2/\eps$ points spaced evenly according to their sorted order.
\end{proof}

\paragraph{Multivariate $\eps$-quantizations.}
We can construct $k$-variate $\eps$-quantizations similarly using the same basic procedure as in Algorithm \ref{alg:rand-draw}.
The output $V_i$ of $f$ is now $k$-variate and thus results in a $k$-dimensional point.
As a result, the reduction of the final size of the point set requires more advanced procedures.

\begin{theorem}
Given a distribution $\mu_P$ of $n$ points, with success probability at least $1-\delta$, we can construct a $k$-variate $\eps$-quantization for $F_{\mu_P}$
of size $O((k/\eps^2)(k +  \log(1/\delta)))$ and in time $O(T_f(n) (1/\eps^2)(k + \log(1/ \delta)))$. 
\label{thm:k-var-q}
\end{theorem}
\begin{proof}
Let $\Eu{R}_+$ describe the family of ranges where a range $A_p = \{q \in \R^k \mid q \preceq p\}$.
In the $k$-variate case there exists a function $g : \R^k \to \R^+$ such that $F_{\mu_P}(v) = \int_{x \preceq v} g(x) \; dx$ where $\int_{x \in \R^k} g(x) \; dx = 1$.  Thus $g$ describes the probability distribution of the values of $f$, given inputs drawn randomly from $\mu_P$.  Hence a random point set $Q$ from $\mu_P$, evaluated as $f(Q)$, is still a random sample from the $k$-variate distribution described by $g$.  Thus, with probability at least $1-\delta$, a set of $O((1/\eps^2)(k + \log (1/\delta)))$ such samples is an $\eps$-sample of $(g,\Eu{R}_+)$, which has VC-dimension $k$, and the samples are also a $k$-variate $\eps$-quantization of $F_{\mu_P}$.  Again, this specific VC-dimension sampling result can also be achieved through a result of Kiefer and Wolfowitz~\cite{KW58}.  
\end{proof}

We can then reduce the size of the $\eps$-quantization $R$ to $O((k^2/\eps) \log^{2k} (1/\eps))$ in $O(|R| (k/\eps^3) \log^{6k} (1/\eps))$ time \cite{Phi08} or to $O((k^2/\eps^2) \log (1/\eps))$ in $O(|R|(k^{3k}/\eps^{2k}) \cdot \log^k(k/\eps))$ time \cite{CM96}, since the VC-dimension is $k$ and each data point requires $O(k)$ storage.

Also on $k$-variate statistics, we can query the resulting $k$-dimensional distribution using other shapes with bounded VC-dimension $\nu$, and if the sample size is $m = O((1/\eps^2)(\nu + \log(1/\delta)))$, then all queries have at most $\eps$-error with probability at least $1-\delta$.  In contrast to the two above results, this statement seems to require the VC-dimension view, as opposed to appealing to the Kiefer-Wolfowitz line of work~\cite{DKW56,KW58}.

\subsection{$(\eps, \delta, \alpha)$-Kernels}
\label{sec:ea-kernel}

The above construction works for a fixed family of summarizing shapes.  In this section, we show how to build a single data structure, an $(\eps, \delta, \alpha)$-kernel, for a distribution $\mu_P$ in $\R^{dn}$ that can be used to construct $(\eps, \alpha)$-quantizations for several families of summarizing shapes. 
This added generality does come at an increased cost in construction.
In particular, an $(\eps, \delta, \alpha)$-kernel of $\mu_P$ is a data structure such that in any query direction $u \in \b{S}^{d-1}$, with probability at least $1-\delta$, we can create an $(\eps, \alpha)$-quantization for the cumulative density function of $\wid(\cdot, u)$, the width in direction $u$.

We follow the randomized framework described above as follows.
The desired $(\eps, \delta, \alpha)$-kernel $\Eu{K}$ consists of a set of $m = O((1/\eps^2) \log (1/\delta))$ $(\alpha/2)$-kernels, $\{K_1, K_2,$ $\ldots, K_m\}$, where each $K_j$ is an $(\alpha/2)$-kernel of a point set $Q_j$ drawn randomly from $\mu_P$.  Given $\Eu{K}$, with probability at least $1-\delta$, we can create an $(\eps, \alpha)$-quantization for the cumulative density function of width over $\mu_P$ in any direction $u \in \b{S}^{d-1}$.
Specifically, let $M = \{\wid(K_j, u) \}_{j=1}^m$.

\begin{lemma}
With probability at least $1-\delta$,
$M$ is an $(\eps, \alpha)$-quantization for the cumulative density function of the width of $\mu_P$ in direction $u$.
\end{lemma}
\begin{proof}
The width $\wid(Q_j, u)$ of a random point set $Q_j$ drawn from $\mu_P$ is a random sample from the distribution over widths of $\mu_P$ in direction $u$.  Thus, with probability at least $1-\delta$, $m$ such random samples would create an $\eps$-quantization.  Using the width of the $\alpha$-kernels $K_j$ instead of $Q_j$ induces an error on each random sample of at most $2\alpha \cdot \wid(Q_j, u)$.
Then for a query width $w$, say there are $\gamma m$ point sets $Q_j$ that have width at most $w$ and $\gamma^\prime m$ $\alpha$-kernels $K_j$ with width at most $w$; see Figure \ref{fig:ea-quant}.
Note that $\gamma^\prime>\gamma$.
Let $\hat{w} = w - 2\alpha w$. For each point set $Q_j$ that has width greater than $w$ but the corresponding $\alpha$-kernel $K_j$ has width at most $w$,
it follows that $K_j$ has width greater than $\hat{w}$.
Thus the number of $\alpha$-kernels $K_j$ that have width at most $\hat{w}$ is at most $\gamma m$, and thus there is a width $w^\prime$ between $w$ and $\hat{w}$ such that the number of $\alpha$-kernels at most $w^\prime$ is exactly $\gamma m$.
\end{proof}

\begin{figure}[htb!]
\begin{center}
\includegraphics{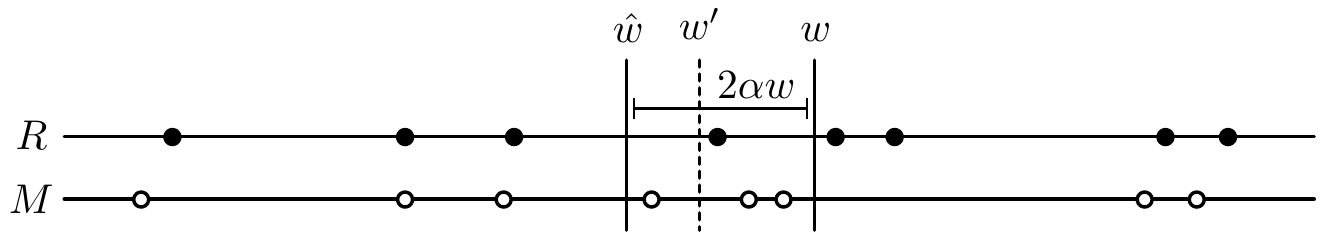}
\end{center}
\caption[Error bounds in $(\eps, \alpha)$-quantization]{\label{fig:ea-quant} $(\eps,\alpha)$-quantization $M$ (white circles) and $\eps$-quantization $R$ (black circles) given a query width $w$.}
\end{figure}

Since each $K_j$ can be computed in $O(n + 1/\alpha^{d-3/2})$ time, we obtain:

\begin{theorem}
We can construct an $(\eps, \delta, \alpha)$-kernel for $\mu_P$ on $n$ points in $\R^d$ of size $O((1/\alpha^{(d-1)/2}) (1/\eps^2) \cdot \log (1/\delta))$ in $O((n + 1/\alpha^{d-3/2}) \cdot (1/\eps^2) \log (1/\delta))$ time.
\end{theorem}

The notion of $(\eps, \alpha)$-quantizations and $(\eps, \delta, \alpha)$-kernels can be extended to $k$-dimensional queries or for a series of up to $k$ queries which all have approximation guarantees with probability $1-\delta$.

\paragraph{Other coresets.}
In a similar fashion, coresets of a point set distribution $\mu_P$ can be formed using coresets for other problems on discrete point sets.  For instance, sample $m = O((1/\eps^2) \log (1/\delta))$ points sets $\{P_1, \ldots, P_m\}$ each from $\mu_P$ and then store $\alpha$-samples $\{Q_1 \subseteq P_1, \ldots, Q_m \subseteq P_m\}$ of each.
When we use random sampling in the second set, then not all distributions $\mu_{p_i}$ need to be sampled for each $P_j$ in the first round.  
This results in an $(\eps, \delta, \alpha)$-sample of $\mu_P$, and can, for example, be used to construct (with probability $1-\delta$) an $(\eps,\alpha)$-quantization for the fraction of points expected to fall in a query disk.
Similar constructions can be done for other coresets, such as $\eps$-nets~\cite{HW87}, $k$-center~\cite{APV02}, or smallest enclosing ball~\cite{BC03}.

\subsection{Measuring the Error}

We have established asymptotic bounds of $m = O((1/\eps^2)(\nu + \log (1/\delta))$ random samples for constructing $\eps$-quantizations.  Now we empirically demonstrate that the constant hidden by the big-O notation is approximately $0.5$, indicating that these algorithms are indeed quite practical.

As a data set, we consider a set of $n = 50$ sample points in $\b{R}^3$ chosen randomly from the boundary of a cylinder piece of length $10$ and radius $1$.  We let each point represent the center of 3-variate Gaussian distribution with standard deviation $2$ to represent the probability distribution of an uncertain point.  This set of distributions describes an uncertain point set $\mu_P : \b{R}^{3n} \to \b{R}^+$.

We want to estimate three statistics on $\mu_P$:
$\dwid$, the width of the points set in a direction that makes an angle of $75^\circ$ with the cylinder axis;
$\diameter$, the diameter of the point set;
and $\seb$, the radius of the smallest enclosing ball (using code from Bernd G\"{a}rtner~\cite{Gar99}).
We can create $\eps$-quantizations with $m$ samples from $\mu_P$, where the value of $m$ is from the set $\{16, 64, 256, 1024, 4096\}$.

We would like to evaluate the $\eps$-quantizations versus the ground truth function $F_{\mu_P}$; however, it is not clear how to evaluate $F_{\mu_P}$.  Instead, we create another $\eps$-quantization $Q$ with $\eta = 100 000$ samples from $\mu_P$, and treat this as if it were the ground truth.
To evaluate each sample $\eps$-quantization $R$ versus $Q$ we find the maximum deviation (i.e. $d_\infty(R,Q) = \max_{q \in \b{R}} |h_R(q) - h_Q(q)|$) with $h$ defined with respect to $\diameter$ or $\dwid$.  This can be done by for each value $r \in R$ evaluating $|h_R(r) - h_Q(r)|$ and $|(h_R(r) - 1/|R|) - h_Q(r)|$ and returning the maximum of both values over all $r \in R$.  

Given a fixed ``ground truth'' quantization $Q$ we repeat this process for $\tau= 500$ trials of $R$, each returning a $d_\infty(R,Q)$ value.  The set of these $\tau$ maximum deviations values results in another quantization $S$ for each of $\diameter$ and $\dwid$, plotted in Figure \ref{fig:maxerr-dwr}.
Intuitively, the maximum deviation quantization $S$ describes the sample probability that $d_\infty(R,Q)$ will be less than some query value.

\begin{figure}[tbh]
\begin{center}
\includegraphics[width=\linewidth]{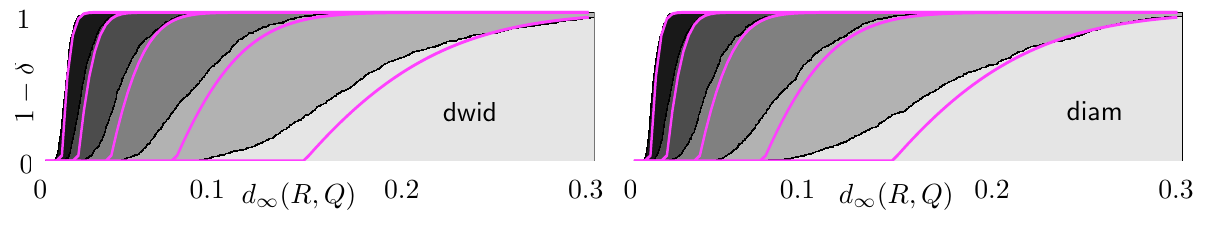}
\end{center}
\caption{\label{fig:maxerr-dwr}
Shows quantizations of $\tau=500$ trials for $d_\infty(R,Q)$ where $Q$ and $R$ measure $\dwid$ and $\diameter$.  The size of each $R$ is $m = \{16, 64, 256, 1024, 4096\}$ (from right to left) and the ``ground truth'' quantization $Q$ has size $\eta = 100 000$.  Smooth, thick curves are $1-\delta = 1-\exp(-2 m \eps^2 + 1)$ where $\eps = d_\infty(R,Q)$.  }
\end{figure}

Note that the maximum deviation quantizations $S$ are similar for both statistics (and others we tried), and thus we can use these plots to estimate $1-\delta$, the sample probability that $d_\infty(R,Q) \leq \eps$, given a value $m$.  We can fit this function as approximately $1-\delta = 1-\exp(-m \eps^2 / C + \nu)$ with $C=0.5$ and $\nu=1.0$.  Thus solving for $m$ in terms of $\eps$, $\nu$, and $\delta$ reveals: $m = C (1/\eps^2) (\nu + \log (1/\delta))$.
This indicates the big-O notation for the asymptotic bound of $O((1/\eps^2) (\nu + \log (1/\delta))$ \cite{LLS01} for $\eps$-samples only hides a constant of approximately $0.5$.

We also ran these experiments to $k$-variate quantizations by considering the width in $k$ different directions.  As expected, the quantizations for maximum deviation can be fit with an equation $1-\delta = 1-\exp(-m\eps^2/C + k)$ with $C=0.5$, so $m \leq C (1/\eps^2)(k + \log 1/\delta)$.  For $k > 2$, this bound for $m$ becomes too conservative; even fewer samples were needed.

\section {Deterministic Computations on Indecisive Point Sets}
\label {sec:pointsetsets}
  
In this section, we take as input a set of $n$ indecisive points, and describe deterministic exact algorithms for creating quantizations of classes of functions on this input.  We characterize problems when these deterministic algorithms can or can not be made efficient.  

\subsection {Polynomial Time Algorithms}
\label {sec:pta}

We are interested in the distribution of the value $f(Q)$ for each \pset $Q \Subset \Eu P$.  Since there are $k^n$ possible \psets, in general we cannot hope to do anything faster than that without making additional assumptions about $f$.
Define $\tilde f (\Eu P, r)$ as the fraction (measured by weight) of \psets of $\Eu P$ for which $f$ gives a value smaller than or equal to $r$.
In this version, for simplicity, we assume general position and that $k^n$ can be described by $O(1)$ words, (handled otherwise in Appendix \ref{sec:bignum}).  
First, we will let $f (Q)$ denote the radius of the smallest enclosing disk of $Q$ in the plane, and show how to solve the decision problem in polynomial time in that case. We then show how to generalize the ideas to other classes of measures.

\paragraph{Smallest enclosing disk.}
Consider the problem where $f$ measures the radius of the smallest enclosing disk of a \pset and let all weights be uniform so $w(q_{i,j}) = 1$ for all $i$ and $j$.  
Evaluating $\tilde f(\Eu P, r)$ in time polynomial in $n$ and $k$ is not completely trivial since there are $k^n$ possible \psets.  However, we can make use of the fact that each smallest enclosing disk is in fact defined by a set of at most $3$ points that lie on the boundary of the disk. 
For each \pset $Q \Subset \Eu P$ we define $B_Q \subseteq Q$ to be this set of at most $3$ points, which we call the \emph {basis} for $Q$.
Bases have the property that $f(Q) = f(B_Q)$.

Now, to avoid having to test an exponential number of \psets, we define a \emph{potential basis} to be a set of at most $3$ points in $\Eu P$ such that each point is from a different $P_i$. Clearly, there are at most $(nk)^3$ possible potential bases, and each \pset $Q \Subset \Eu P$ has one as its basis.
Now, we only need to count for each potential basis the number of \psets it represents.
Counting the number of samples that have a certain basis is easy for the smallest enclosing circle. Given a basis $B$, we count for each indecisive point $P$ that does not contribute a point to $B$ itself how many of its members lie inside the smallest enclosing circle of $B$, and then we multiply these numbers.
\tweeplaatjes {sec-basis} {sec-basis-count} {(a) The smallest enclosing circle of a set of points is defined by two or three points on the boundary. (b) This circle contains one purple (dark) point, four blue (medium) points, and two yellow (light) points. Hence there are $1 \times 4 \times 2 = 8$ samples that have this basis.}
Figure~\ref {fig:sec-example+sec-example-smallest+sec-example-largest+sec-example-result} illustrates the idea.

\vierplaatjes {sec-example} {sec-example-smallest} {sec-example-largest} {sec-example-result}
{
  (a) Example input with $n = 3$ and $k = 3$. 
  (b) One possible basis, consisting of $3$ points. This basis has one support: the basis itself. 
  (c) Another possible basis, consisting of $2$ points. This basis has three \psets. 
  (d) The graph showing for each diameter $d$ how many \psets do not exceed that diameter. This corresponds to the cumulative distribution of the radius of the smallest enclosing disk of these points.
}

Now, for each potential basis $B$ we have two values: the number of \psets that have $B$ as their basis, and the value $f(B)$. We can sort these $O((nk)^3)$ pairs on the value of $f$, and the result provides us with the required distribution.
We spend $O(nk)$ time per potential basis for counting the points inside and $O(n)$ time for multiplying these values,
so combined with $O((nk)^3)$ potential bases this gives $O((nk)^4)$ total time.

\begin {theorem}
Let $\Eu P$ be a set of $n$ sets of $k$ points.
In $O ((nk)^4)$ time, we can compute a data structure of $O ((nk)^3)$ size that can tell us in $O(\log (nk))$ time for any value $r$ how many \psets of $Q \Subset \Eu P$ satisfy $f(Q) \leq r$.
\end {theorem}

\paragraph{LP-type problems.}

The approach described above also works for measures $f : \Eu P \to \R$ other than the smallest enclosing disk.  In particular, it works for LP-type problems~\cite{SW92} that have constant combinatorial dimension.
An \emph{LP-type problem} provides a set of constraints $H$ and a function $\omega : 2^H \to \R$ with the following two properties:

\begin{tabular} {rl} 
\textsc{Monotonicity:} & For any $F \subseteq G \subseteq H$, $\omega(F) \leq \omega(G)$.\\
\textsc{Locality:} & For any $F \subseteq G \subseteq H$ with $\omega(F) = \omega(G)$\\ &and an $h \in H$ such that $\omega(G \cup h) > \omega(G)$\\ &implies that $\omega(F \cup h) > \omega(F)$.\\
\end{tabular}

A \emph{basis} for an LP-type problem is a subset $B \subset H$ such that $\omega(B') < \omega(B)$ for all proper subsets $B'$ of $B$.  And we say that $B$ is a basis for a subset $G \subseteq H$ if $B \subseteq G$, $\omega(B) = \omega(G)$ and $B$ is a basis.  A constraint $h \in H$ \emph{violates} a basis $B$ if $w(B \cup h) > w(B)$.
The radius of the smallest enclosing ball is an LP-type problem (where the points are the constraints and $\omega(\cdot) = f(\cdot)$) as are linear programming and many other geometric problems.
Let the maximum cardinality of any basis be the \emph{combinatorial dimension} of a problem.  

For our algorithm to run efficiently, we assume that our LP-type problem has available the following algorithmic primitive, which is often assumed for LP-type problems with constant combinatorial dimension~\cite{SW92}.  For a subset $G \subset H$ where $B$ is known to be the basis of $G$ and a constraint $h \in H$, a \emph{violation test} determines in $O(1)$ time if $\omega(B \cup h) > \omega(B)$; i.e., if $h$ \emph{violates} $B$.  More specifically, given an efficient violation test, we can ensure a stronger algorithmic primitive.  A \emph{full violation test} is given a subset $G \subset H$ with known basis $B$ and a constraint $h \in H$ and determines in $O(1)$ time if $\omega(B) < \omega(G \cup h)$.  This follows because we can test in $O(1)$ time if $\omega(B) < \omega(B \cup h)$; \textsc{monotonicity} implies that $\omega(B) < \omega(B \cup h)$ only if $\omega(B) < \omega(B \cup h) \leq \omega(G \cup h)$, and \textsc{locality} implies that $\omega(B) = \omega(B \cup h)$ only if $\omega(B) = \omega(G) = \omega(G \cup h)$.
Thus we can test if $h$ violates $G$ by considering just $B$ and $h$, but if either \textsc{monotonicity} or \textsc{locality} fail for our problem we cannot.

We now adapt our algorithm to LP-type problems where elements of each $P_i$ are potential constraints and the ranking function is $f$.
When the combinatorial dimension is a constant $\beta$, we need to consider only $O((nk)^\beta)$ bases, which will describe all possible supports.

The full violation test implies that given a basis $B$, we can measure the sum of probabilities of all \psets of $\Eu P$ that have $B$ as their basis in $O(nk)$ time.  For each indecisive point $P$ such that $B \cap P = \emptyset$, we sum the probabilities of all elements of $P$ that do not violate $B$.  The product of these probabilities times the product of the probabilities of the elements in the basis, gives the probability of $B$ being the true basis.  See Algorithm \ref{alg:count-prob} where the indicator function applied $1(f(B \cup \{p_j\}) = f(B))$ returns $1$ if $p_j$ does not violate $B$ and $0$ otherwise.
It runs in $O((nk)^{\beta+1})$ time.

\begin{algorithm}[h!!t]
\caption{\label{alg:count-prob}Construct Probability Distribution for $f(\Eu P)$.}
\begin{algorithmic}[1]
\FOR {all potential bases $B \subset Q \Subset \Eu P$}
  \FOR {$i = 1$ \textbf{to} $n$}
    \IF {there is a $j$ such that $p_{ij} \in B$}
      \STATE Set $w_i = w(p_{ij})$.
    \ELSE
      \STATE Set $w_i = \sum_{j=1}^k w(p_{ij}) 1(f(B \cup \{p_j\}) = f(B))$.
    \ENDIF
  \ENDFOR
  \STATE Store a point with value $f(B)$ and weight $(1/k^n) \prod_i w_i$.
\ENDFOR
\end{algorithmic}
\end{algorithm}

As with the special case of smallest enclosing disk, we can create a distribution over the values of $f$ given an indecisive point set $\Eu P$.  For each basis $B$ we calculate $\mu(B)$, the summed probability of all \psets that have basis $B$, and $f(B)$. We can then sort these pairs according to the value as $f$ again.
For any query value $r$, we can retrieve $\tilde f(\Eu P, r)$ in $O(\log (nk))$ time and it takes $O(n)$ time to describe (because of its long length).

\begin{theorem}
Given a set $\Eu P$ of $n$ indecisive point sets of size $k$ each, and given an LP-type problem $f : \Eu P \to \R$ with combinatorial dimension $\beta$, we can create the distribution of $f$ over $\Eu P$ in $O((nk)^{\beta+1})$ time.  The size of the distribution is $O(n (nk)^\beta)$.
\label{thm:gen-comb-dist}
\end{theorem}

If we assume general position of $\Eu P$ relative to $f$, then we can often slightly improve the runtime needed to calculate $\mu(B)$ using range searching data structures.  However, to deal with degeneracies, we may need to spend $O(nk)$ time per basis, regardless.

If we are content with an approximation of the distribution rather than an exact representation, then it is often possible to drastically reduce the storage and runtime following techniques discussed in Section \ref{sec:algQ}.  

Measures that fit in this framework for points in $\R^d$ include smallest enclosing axis-aligned rectangle (measured either by area or perimeter) ($\beta = 2d$), smallest enclosing ball in the $L_1$, $L_2$, or $L_\infty$ metric ($\beta = d+1$), directional width of a set of points ($\beta = 2$), and, after dualizing, linear programming ($\beta = d$).
These approaches also carry over naturally to deterministically create polynomial-sized $k$-variate quantizations.

\subsection {Hardness Results}

In this section, we examine some extent measures that do not fit in the above framework.
First, diameter does not satisfy the \textsc{locality} property, and hence we cannot efficiently perform the full violation test.  
We show that a decision variant of diameter is \#P-Hard, even in the plane, and thus (under the assumption that \#P $\neq$ P), there is no polynomial time solution.  This result holds despite the fact that diameter has a combinatorial dimension of $2$, implying that the associated quantization has at most $O((nk)^2)$ steps.  
Second, the area of the convex hull does not have a constant combinatorial dimension, thus we can show the resulting distribution may have exponential size.

\paragraph{Diameter.}
The \emph {diameter} of a set of points in the plane is the largest distance between any two points. We will show that the counting problem of computing $\tilde f (\Eu P, r)$ is \#P-hard when $f$ denotes the diameter.

        \begin{problem}
        \PDIAM: 
        Given a parameter $d$ and a set $\Eu{P} = \{P_1, \ldots, P_n\}$ of $n$ sets, each consisting of $k$ points in the plane, how many \psets $Q \Subset \Eu P$ have $f(Q) \leq d$?
        \label{prb:count-diam}
        \end{problem}

We will now prove that Problem \ref{prb:count-diam} is \#P-hard.
Our proof has three steps. We first show a special version of \twoSAT has a polynomial reduction from Monotone \twoSAT, which is \#P-complete~\cite{Val79}.
Then, given an instance of this special version of \twoSAT, we construct a graph with weighted edges on which the diameter problem is equivalent to this \twoSAT instance.  Finally, we show the graph can be embedded as a straight-line graph in the plane as an instance of \PDIAM.

Let 3CLAUSE-\twoSAT be the problem of counting the number of solutions to a 2SAT formula, where each variable occurs in at most three clauses, and each variable is either in exactly one clause or is negated in exactly one clause.  Thus, each distinct literal appears in at most two clauses.  

\begin {lemma}
\label{lem:3CLAUSE-2SAT}
Monotone \twoSAT has a polynomial reduction to 3CLAUSE-\twoSAT.
\end {lemma}
\begin{proof}
The Monotone \twoSAT problem counts the number satisfying assignments to a \twoSAT instance where each clause has at most two variables and no variables are negated.  
Let $X = \{(x, y_1), (x, y_2), \ldots, (x, y_u)\}$ be the set of $u$ clauses which contain variable $x$ in a Monotone \twoSAT instance.  We replace $x$ with $u$ variables $\{z_1, z_2, \ldots, z_u\}$ and we replace $X$ with the following $2u$ clauses $\{(z_1, y_1), (z_2, y_2), \ldots, (z_u, y_u)\}$ and $\{(z_1, \neg z_2), (z_2, \neg z_3), \ldots,$ $(z_{u-1}, \neg z_u), (z_u, \neg z_1)\}$.
The first set of clauses preserves the relation with other original variables and the second set of clauses ensures that all of the new variables have the same value (i.e. \textsf{TRUE} or \textsf{FALSE}).
This procedure is repeated for each original variable that is in more than $1$ clause.
\end{proof}

We convert this problem into a graph problem by, for each variable $x_i$, creating a set $P_i = \{p_i^+, p_i^-\}$ of two points.  Let $S = \bigcup_i P_i$.  Truth assignments of variables correspond to a \pset as follows.  If $x_i$ is set \textsf{TRUE}, then the \pset includes $p_i^+$, otherwise the \pset includes $p_i^-$.  We define a distance function $f$ between points, so that the distance is greater than $d$ (long) if the corresponding literals are in a clause, and less than $d$ (short) otherwise.  If we consider the graph formed by only long edges, we make two observations.
First, the maximum degree is $2$, since each literal is in at most two clauses.
Second, there are no cycles since a literal is only in two clauses if in one clause the other variable is negated, and negated variables are in only one clause.
These two properties imply we can use the following construction to show that the \PDIAM problem is as hard as counting Monotone \twoSAT solutions, which is \#P-complete.

\begin{lemma}
\label{lem:embed-DIAM}
An instance of \PDIAM reduced from 3CLAUSE-\twoSAT can be embedded so $\Eu{P} \subset \R^2$.
\end{lemma}
\begin{proof}

Consider an instance of 3CLAUSE-\twoSAT where there are $n$ variables, and thus the corresponding graph has $n$ sets $\{P_i\}_{i=1}^n$.  We construct a sequence $\Gamma$ of $n' \in [2n,4n]$ points.
It contains all points from $\Eu{P}$ and a set of at most as many dummy points.
First organize a sequence $\Gamma'$ so if two points $q$ and $p$ have a long edge, then they are consecutive.  Now for any pair of consecutive points in $\Gamma'$ which do not have a long edge, insert a dummy point between them to form the sequence $\Gamma$.  Also place a dummy point at the end of $\Gamma$.

We place all points on a circle $C$ of diameter $d/\cos(\pi / n')$, see Figure \ref{fig:embed-diam}.
We first place all points on a semicircle of $C$ according to the order of $\Gamma$, so each consecutive points are $\pi/n'$ radians apart.  Then for every other point (i.e. the points with an even index in the ordering $\Gamma$) we replace it with its antipodal point on $C$, so no two points are within $2\pi/n'$ radians of each other.  Finally we remove all dummy points.
This completes the embedding of $\Eu{P}$, we now need to show that only points with long edges are further than $d$ apart.

\eenplaatje {embed-diam}
{Embedded points are solid, at center of circles of radius $d$.  Dummy points hollow. Long edges are drawn between points at distance greater than $d$.}

We can now argue that only vertices which were consecutive in $\Gamma$ are further than $d$ apart, the remainder are closer than $d$.  Consider a vertex $p$ and a circle $C_p$ of radius $d$ centered at $p$.  Let $p'$ be the antipodal point of $p$ on $C$.  $C_p$ intersects $C$ at two points, at $2\pi/n'$ radians in either direction from $p'$.  Thus only points within $2\pi/n'$ radians of $p'$ are further than a distance $d$ from $p$.  This set includes only those points which are adjacent to $p$ in $\Gamma$, which can only include points which should have a long edge, by construction.  
\end{proof}

Combining Lemmas~\ref {lem:3CLAUSE-2SAT} and \ref {lem:embed-DIAM}: 
\begin{theorem}
\PDIAM is \#P-hard.
\end{theorem}

\paragraph{Convex hull.}
Our LP-type framework also does not work for any properties of the convex hull (e.g. area or perimeter) because it does not have constant combinatorial dimension; a basis could have size $n$.   In fact, the complexity of the distribution describing the convex hull may be $\Omega(k^n)$, since if all points in $\Eu P$ lie on or near a circle, then every \pset $Q \Subset \Eu P$ may be its own basis of size $n$, and have a different value $f(Q)$.

\section {Deterministic Algorithms for Approximate Computations on Uncertain Points}
\label {sec:distributions}

In this section we show how to approximately answer questions about most representations of independent uncertain points; in particular, we handle representations that have almost all ($1-\eps$ fraction) of their mass with bounded support in $\b{R}^d$ and is described in a compact manner (see Appendix \ref{sec:eps-dist}).  
Specifically, in this section, we are given a set $\Eu P = \{P_1, P_2, P_3, \ldots, P_n\}$ of $n$ independent random variables over the universe $\R^d$, together with a set $\Eu \mu_{\Eu{P}} = \{\mu_1, \mu_2, \mu_3, \ldots, \mu_n\}$ of $n$ probability distributions that govern the variables, that is, $X_i \sim \mu_i$.
Again, we call a set of points $Q = \{q_1, q_2, q_3, \ldots, q_n\}$ a \pset of $\Eu P$, and because of the independence we have probability $Pr[\Eu P = Q] = \prod_{i} Pr[P_i = p_i]$.

    The main strategy will be to replace each distribution $\mu_i$ by a discrete point set $P_i$, such that the uniform distribution over $P_i$ is ``not too far'' from $\mu_i$ ($P_i$ is not the most obvious $\eps$-sample of $\mu_i$). Then we apply the algorithms from Section~\ref {sec:pointsetsets} to the resulting set of point sets. Finally, we argue that the result is in fact an $\eps$-quantization of the distribution we are interested in.   Using results from Section \ref{sec:rand-eQ} we can simplify the output in order to decrease the space complexity for the data structure, without increasing the approximation factor too much.

\paragraph{General approach.}
Given a distribution $\mu_i : \R^2 \to \R^+$ describing uncertain point $P_i$ and a function $f$ of bounded combinatorial dimension $\beta$ defined on a \pset of $\Eu P$, we can describe a straightforward range space $T_i = (\mu_i, \Eu A_f)$, where $\Eu A_f$ is the set of ranges corresponding to the bases of $f$ (e.g., when $f$ measures the radius of the smallest enclosing ball, $\Eu A_f$ would be the set of all balls).
More formally, $\Eu A_f$ is the set of subsets of $\R^d$ defined as follows: for every set of $\beta$ points which define a basis $B$ for $f$, $\Eu A_f$ contains a range $A$ that contains all points $p$ such that $f(B) = f(B \cup \{p\})$.
However, taking $\eps$-samples from each $T_i$ is \emph{not} sufficient to create sets $Q_i$ such that $\Eu Q = \{Q_1, Q_2, \ldots, Q_n\}$ so for all $r$ we have $|\tilde f(\Eu P, r) - \tilde f(\Eu Q, r) | \leq \eps$.

$\tilde f(\Eu P, r)$ is a complicated joint probability depending on the $n$ distributions and $f$, and the $n$ straightforward $\eps$-samples do not contain enough information to decompose this joint probability.  
The required $\eps$-sample of each $\mu_i$ should model $\mu_i$ in relation to $f$ and any instantiated point $p_i$ representing $\mu_j$ for $i\neq j$.  
The following crucial definition allows for the range space to depend on any $n-1$ points, including the possible locations of each uncertain point.  

Let $\Eu{A}_{f,n}$ describe a family of Lebesgue-measurable sets defined by $n-1$ points $Z \subset \R^d$ and a value $w$.  Specifically, $A(Z,w) \in \Eu{A}_{f,n}$ is the set of points $\{p \in \R^d \mid f(Z \cup p) \leq w\}$.
We describe examples of $\Eu A_{f,n}$ in detail shortly, but first we state the key theorem using this definition.  Its proof, delayed until after examples of $\Eu A_{f,n}$, will make clear how $(\mu_i, \Eu{A}_{f,n})$ exactly encapsulates the right guarantees to approximate $\hat f(\Eu P, r)$, and thus why $(\mu_i, \Eu A_f)$ does not.

\begin{theorem} \label{thm:n-apx}
Let $\Eu P = \{P_1, \ldots, P_n\}$ be a set of uncertain points where each $P_i \sim \mu_i$.  
For a function $f$, let $Q_i$ be an $\eps'$-sample of $(\mu_i, \Eu A_{f,n})$ and let $\Eu Q = \{Q_1, \ldots, Q_n\}$.
Then for any $r$, 
$\left| \hat f(\Eu P, r)  - \tilde f(\Eu Q, r) \right| \leq  \eps^\prime n.$
\end{theorem}

\paragraph{Smallest axis-aligned bounding box by perimeter.}
Given a set of points $P \subset \R^2$, let $f(P)$ represent the perimeter of the smallest axis-aligned box that contains $P$.  Let each $\mu_i$ be a bivariate normal distribution with constant variance.  Solving $f(P)$ is an LP-type problem with combinatorial dimension $\beta = 4$, and as such, we can describe the basis $B$ of a set $P$ as the points with minimum and maximum $x$- and $y$-coordinates.  Given any additional point $p$, the perimeter of size $\rho$ can only be increased to a value $w$ by expanding the range of $x$-coordinates, $y$-coordinates, or both.  As such, the region of $\R^2$ described by a range $A(P, w) \in \Eu A_{f,n}$ is defined with respect to the bounding box of $P$ from an edge increasing the $x$-width or $y$-width by $(w-\rho)/2$, or from a corner extending so the sum of the $x$ and $y$ deviation is $(w-\rho)/2$.  
See Figure \ref{fig:shape-aabbp}.

Since any such shape defining a range $A(P,w) \in \Eu A_{f,n}$ can be described as the intersection of $k=4$ slabs along fixed axis (at $0^\circ$, $45^\circ$, $90^\circ$, and $135^\circ$), we can construct an $(\eps/n)$-sample $Q_i$ of $(\mu_i, \Eu A_{f,n})$ of size $k = O((n/\eps) \log^{8} (n/\eps))$ in $O((n^6/\eps^6) \log^{27} (n/\eps))$ time~\cite{Phi08}.  From Theorem \ref{thm:n-apx}, it follows that for $\Eu Q = \{Q_1, \ldots, Q_n\}$ and any $r$ we have $\left| \hat f(\Eu X, r) - \tilde f(\Eu Q, r) \right| \leq \eps$.

We can then apply Theorem \ref{thm:gen-comb-dist} to build an $\eps$-quantization of $f(\Eu X)$ in $O((nk)^5) 
= O(((n^2/\eps) \log^8 (n/\eps))^5) 
= O((n^{10}/\eps^5) \log^{40} (n/\eps))$ time.  
The size can be reduced to $O(1/\eps)$ within that time bound. 

\begin{corollary}
Let $\Eu P = \{P_1, \ldots, P_n\}$ be a set of indecisive points where each $P_i \sim \mu_i$ is bivariate normal with constant variance.  Let $f$ measure the perimeter of the smallest enclosing axis-aligned bounding box.
We can create an $\eps$-quantization of $f(\Eu P)$ in $O((n^{10}/\eps^5) \log^{40} (n/\eps))$ time  of size $O(1/\eps)$.
\label{cor:n-apx-aabbp}
\end{corollary}

\drieplaatjes {shape-aabbp} {shape-sebl2} {shape-sebl2-wedges}
{
  (a) A shape from $\Eu{A}_{f,n}$ for axis-aligned bounding box, measured by perimeter.  
  (b) A shape from $\Eu{A}_{f,n}$ for smallest enclosing ball using the $L_2$ metric in $\R^2$. The curves are circular arcs of two different radii. 
  (c) The same shape divided into wedges from $\Eu{W}_{f,n}$.
}

\paragraph{Smallest enclosing disk.}
Given a set of points $P \subset \R^2$, let $f(P)$ represent the radius of the smallest enclosing disk of $P$.  Let each $\mu_i$ be a bivariate normal distribution with constant variance.  Solving $f(P)$ is an LP-type problem with combinatorial dimension $\beta = 3$, and the basis $B$ of $P$ generically consists of either $3$ points which lie on the boundary of the smallest enclosing disk, or $2$ points which are antipodal on the smallest enclosing disk.  However, given an additional point $p \in \R^2$, the new basis $B_p$ is either $B$ or it is $p$ along with $1$ or $2$ points which lie on the convex hull of $P$.

We can start by examining all pairs of points $p_i, p_j \in P$ and the two disks of radius $w$ whose boundary circles pass through them.  If one such disk $D_{i,j}$ contains $P$, then $D_{i,j} \subset A(P, w) \in \Eu A_{f, |P|+1}$.  For this to hold, $p_i$ and $p_j$ must lie on the convex hull of $P$ and no point that lies between them on the convex hull can contribute to such a disk.  Thus there are $O(n)$ such disks.
We also need to examine the disks created where $p$ and one other point $p_i \in P$ are antipodal.  The boundary of the union of all such disks which contain $P$ is described by part of a circle of radius $2w$ centered at some $p_i \in P$.  Again, for such a disk $B_i$ to describe a part of the boundary of $A(P,w)$, the point $p_i$ must lie on the convex hull of $P$.  The circular arc defining this boundary will only connect two disks $D_{i,j}$ and $D_{k,i}$ because it will intersect with the boundary of $B_j$ and $B_k$ within these disks, respectively.  An example of $A(P,w)$ is shown in Figure \ref{fig:shape-sebl2}.

Unfortunately, the range space $(\R^2, \Eu A_{f,n})$ has VC-dimension $O(n)$; it has $O(n)$ circular boundary arcs.  So, creating an $\eps$-sample of $T_i = (\mu_i, \Eu A_{f,n})$ would take time exponential in $n$.  However, we can decompose any range $A(P, w) \in \Eu A_{f,n}$ into at most $2n$ ``wedges.''  We choose one point $y$ inside the convex hull of $P$.  For each circular arc on the boundary of $A(P, w)$ we create a wedge by coning that boundary arc to $y$.  Let $\Eu W_f$ describe all wedge shaped ranges.  Then $S = (\R^2, \Eu W_f)$ has VC-dimension $\nu_S$ at most $9$ since it is the intersection of $3$ ranges (two halfspaces and one disk) that can each have VC-dimension $3$.
We can then create $Q_i$, an $(\eps/2n^2)$-sample of $S_i = (\mu_i, \Eu W_f)$, of size $k = O((n^4/\eps^2) \log (n/\eps))$ in $O((n^2/\eps)^{5+2\cdot9} \log^{2+9} (n/\eps)) = O((n^{46}/\eps^{23}) \log^{11}(n/\eps))$ time, via Corollary \ref{cor:norm-es} (Appendix \ref{sec:eps-dist}).
It follows that $Q_i$ is an $(\eps/n)$-sample of $T_i = (\mu_i, \Eu A_{f,n})$, since any range $A(Z,w) \in \Eu A_{f,n}$ can be decomposed into at most $2n$ wedges, each of which has counting error at most $\eps/2n$, thus the total counting error is at most $\eps$.

Invoking Theorem \ref{thm:n-apx}, it follows that $\Eu Q = \{Q_1, \ldots, Q_n\}$, for any $r$ we have $\left| \hat f(\Eu P, r) - \tilde f(\Eu Q, r) \right| \leq \eps$.   We can then apply Theorem \ref{thm:gen-comb-dist} to build an $\eps$-quantization of $f(\Eu P)$ in $O((nk)^4) = O((n^{20}/\eps^{8}) \log^4 (n/\eps))$ time.   This is dominated by the time for creating the $n$ $(\eps/n^2)$-samples, even though we only need to build one and then translate and scale to the rest.  
Again, the size can be reduced to $O(1/\eps)$ within that time bound.  

\begin{corollary}
Let $\Eu P = \{P_1, \ldots, P_n\}$ be a set of indecisive points where each $P_i \sim \mu_i$ is bivariate normal with constant variance.  Let $f$ measure the radius of the smallest enclosing disk.
We can create an $\eps$-quantization of $f(\Eu P)$ in $O((n^{46}/\eps^{23}) \log^{11} (n/\eps))$ time  of size $O(1/\eps)$.
\label{cor:n-apx-seb2}
\end{corollary}

Now that we have seen two concrete examples, we prove Theorem \ref{thm:n-apx}. More examples can be found in Appendix \ref {app:moreexamples}.

\paragraph{Proof of Theorem \ref{thm:n-apx}.}

When each $P_i$ is drawn from a distribution $\mu_i$, then we can write $\hat f(\Eu P, r)$ as the probability that $f(\Eu P) \leq r$ as follows.  Let $1(\cdot)$ be the indicator function, i.e., it is $1$ when the condition is true and $0$ otherwise.
\[
\hat f(\Eu P, r) =
\int_{p_1} \mu_1(p_1) \ldots \int_{p_n} \mu_n(p_n)
1(f(\{p_1, p_2, \ldots, p_n\}) \leq r) \;
d p_n d p_{n-1} \ldots d p_1
\]
Consider the inner most integral
\[
\int_{p_n} \mu_n(p_n)  1(f(\{p_1, p_2, \ldots, p_n\}) \leq r) \; d p_n,
\]
where $\{p_1, p_2 \ldots, p_{n-1}\}$ are fixed.  The indicator function $1(\cdot) = 1$ is true 
when
$
f(\{p_1, p_2, \ldots, p_{n-1}, p_n\}) \leq r,
$
and hence $p_n$ is contained in a shape $A(\{p_1, \ldots, p_{n-1}\}, r) \in \Eu{A}_{f,n}$.
Thus if we have an $\eps^\prime$-sample $Q_n$ for $(\mu_n,\Eu{A}_{f,n})$, then we can guarantee that
\[
\int_{p_n}  \mu_{n}(p_n)  1(f(\{p_1, p_2, \ldots, p_n\}) \leq r) \; d p_n
\leq
\frac{1}{|Q_n|} \sum_{p_n \in Q_{n}} 1(f(\{p_1, p_2, \ldots, p_{n-1}, p_n\}) \leq r) + \eps^\prime.
\]
We can then move the $\eps^\prime$ outside and change the order of the integrals to write:
\[
\hat f(\Eu X, r) \leq  \frac{1}{|Q_n|} \sum_{p_n \in Q_{n}} 
\left( \int_{p_1}  \mu_1(p_1) ..  \int_{p_{n-1}}  \mu_{n-1}(p_{n-1})
1(f(\{p_1,.., p_n\}) \leq r) 
d p_{n-1} .. d {p_1}\right)  +  \eps^\prime.
\]
Repeating this procedure $n$ times we get:
\[
\hat f(\Eu X, r)
\leq
\left(\prod_{i=1}^n \frac{1}{|Q_i|}\right)
 \sum_{p_1 \in Q_1}  \dots \sum_{p_n \in Q_n}
1(f(\{p_1, \ldots, p_n\}) \leq r)
+ \eps^\prime n
 =
\tilde{f}(\Eu Q,r) + \eps^\prime n,
\]
where $\Eu Q = \bigcup_i Q_i$.

Similarly we can achieve a symmetric lower bound for $\hat f(\Eu X,r)$.  \qed

\section {Shape Inclusion Probabilities}

So far, we have been concerned only with computing probability distributions of single-valued functions on a set of points. However, many geometric algorithms produce more than just a single value. In this section, we consider dealing with uncertainty when computing a two-dimensional shape (such as the convex hull, or MEB) of a set of points directly.

\subsection{Randomized Algorithms}
We can also use a variation of Algorithm \ref{alg:rand-draw} to construct $\eps$-shape inclusion probability functions.
For a point set $Q \subset \R^d$, let the summarizing shape $S_Q = \Eu{S}(Q)$ be from some geometric family $\Eu{S}$ so $(\R^d, \Eu{S})$ has bounded VC-dimension $\nu$.
We randomly sample $m$ point sets $\Eu{Q} = \{Q_1, \ldots, Q_m\}$ each from $\mu_P$ and then find the summarizing shape $S_{Q_j} = \Eu{S}(Q_j)$ (e.g. minimum enclosing ball) of each $Q_j$.  Let this set of shapes be $S^{\Eu{Q}}$.  If there are multiple shapes from $\Eu{S}$ which are equally optimal, choose one of these shapes at random.
For a set of shapes $S^\prime \subseteq \Eu{S}$, let $S^\prime_p \subseteq S^\prime$ be the subset of shapes that contain $p \in \R^d$.
We store $S^{\Eu{Q}}$ and evaluate a query point $p \in \R^d$ by counting what fraction of the shapes the point is contained in, specifically returning $|S^{\Eu{Q}}_p| / |S^{\Eu{Q}}|$ in $O(\nu |S^{\Eu{Q}}|) = O(\nu m)$ time.  In some cases, this evaluation can be sped up with point location data structures.

To state the main theorem most cleanly, for a range space $(\R^d, \Eu{S})$, denote its dual range space as $(\Eu{S}, P^*)$ where $P^*$ is all subsets $\Eu{S}_p \subseteq \Eu{S}$, for any point $p \in \R^d$, such that $\Eu{S}_p = \{S \in \Eu{S} \mid p \in S\}$.  Recall that the VC-dimension $\nu$ of $(\Eu{S}, P^*)$ is at most $2^{\nu^\prime+1}$ where $\nu^\prime$ is the VC-dimension of $(\R^d, \Eu{S})$, but is typically much smaller.  

\begin{theorem}
Consider a family of summarizing shapes $\Eu{S}$ with dual range space $(\Eu{S}, P^*)$ with VC-dimension $\nu$, and where it takes $T_{\Eu{S}}(n)$ time to determine the summarizing shape $\Eu{S}(Q)$ for any point set $Q \subset \R^d$ of size $n$.
For a distribution $\mu_P$ of a point set of size $n$, with probability at least $1-\delta$,
we can construct an $\eps$-\sip function of size $O((\nu/\eps^2) (\nu +  \log (1/\delta)))$
and in time $O(T_{\Eu{S}}(n)(1/\eps^2) \log (1/\delta))$.
\label{thm:rand-sip}
\end{theorem}
\begin{proof}
Using the above algorithm, sample $m = O((1/\eps^2) (\nu + \log (1/\delta)))$ point sets $Q$ from $\mu_P$ and generate the $m$ summarizing shapes $S_Q$.
Each shape is a random sample from $\Eu{S}$ according to $\mu_P$, and thus $S^{\Eu{Q}}$ is an $\eps$-sample of $(\Eu{S}, P^*)$.

Let $w_{\mu_P}(S)$, for $S \in \Eu{S}$, be the probability that $S$ is the summarizing shape of a point set $Q$ drawn randomly from $\mu_P$.
For any $\Eu{S}^\prime \subseteq P^*$, let $W_{\mu_P}(\Eu{S}^\prime) = \int_{S \in \Eu{S}^\prime} w_{\mu_P}(S) \, dS$ be the probability that some shape from the subset $\Eu{S}^\prime$ is the summarizing shape of $Q$ drawn from $\mu_P$.

We approximate the \sip function at $p \in \R^d$ by returning the fraction $|S^{\Eu{Q}}_p| / m$.
The true answer to the \sip function at $p \in \R^d$ is $W_{\mu_P}(\Eu{S}_p)$.
Since $S^{\Eu{Q}}$ is an $\eps$-sample of $(\Eu{S}, P^*)$, then with probability at least $1-\delta$
\[
\left| \frac{|S^{\Eu{Q}}_p|}{m} - \frac{W_{\mu_P}(\Eu{S}_p)}{1} \right|
=
\left| \frac{|S^{\Eu{Q}}_p|}{|S^{\Eu{Q}}|} - \frac{W_{\mu_P}(\Eu{S}_p)}{W_{\mu_P}(P^*)} \right|
\leq \eps.  \qedhere
\]
\end{proof}

\paragraph{Representing $\eps$-\sip functions by isolines.}

\vierplaatjes [height=120pt]
{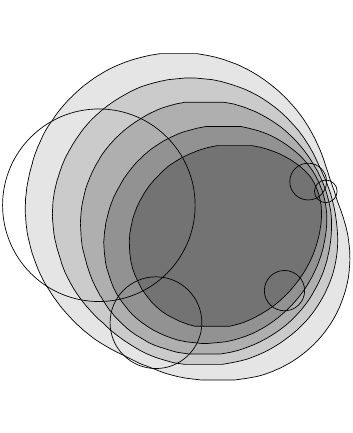} {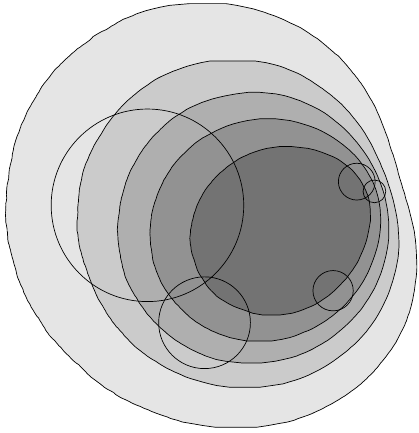} 
{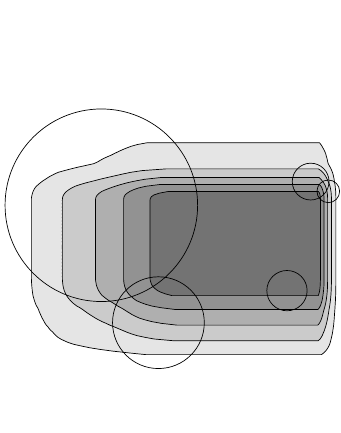} {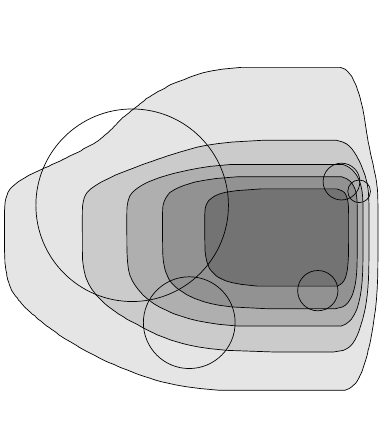}
{ \label {fig:sip-seb} 
  The \sip for the smallest enclosing ball (a,b) or smallest enclosing 
  axis-aligned rectangle (c,d), for uniformly (a,c) or normally (b,d) 
  distributed points. Isolines are drawn for $p\in \{0.1,0.3,0.5,0.7,0.9\}$.
}

Shape inclusion probability functions are density functions.  A convenient way of visually representing a density function in $\R^2$ is by drawing the isolines.  A \emph{$\gamma$-isoline} is a collection of closed curves bounding the regions of the plane where the density function is greater than $\gamma$.

In each part of Figure \ref{fig:sip-seb} a set of 5 circles correspond to points with a probability distribution.   In part (a,c) the probability distribution is uniform over the inside of the circles.  In part (b,d) it is drawn from a normal distribution with standard deviation given by the radius.  We generate $\eps$-\sip functions for the smallest enclosing ball in Figure \ref{fig:sip-seb}(a,b) and for the smallest axis-aligned rectangle in Figure \ref{fig:sip-seb}(c,d).

In all figures we draw approximations of $\{.9,.7,.5,.3,.1\}$-isolines.
These drawing are generated by randomly selecting $m = 5000$ (Figure \ref{fig:sip-seb}(a,b)) or $m=25000$ (Figure \ref{fig:sip-seb}(c,d)) shapes, counting the number of inclusions at different points in the plane and interpolating to get the isolines.
  The innermost and darkest region has probability $> 90\%$, the next one probability $> 70\%$, etc., the outermost region has probability $< 10 \%$.

\subsection {Deterministic Algorithms}

We can also adapt the deterministic algorithms presented in Section \ref{sec:pta} to deterministically create SIP functions for indecisive (or $\eps$-SIPs for many uncertain) points.  
We again restrict our attention to a class of LP-type problems, specifically, problems where the output is a (often minimal) summarizing shape $S(Q)$ of a data set $Q$ where the boundaries of the two shapes $S(Q)$ and $S(Q')$ intersect in at most a constant number of locations.  An example is the smallest enclosing disk in the plane, where the circles on the boundaries of any two disks intersect at most twice.  

As in Section \ref{sec:pta}, this problem has a constant combinatorial dimension $\beta$ and possible locations of indecisive points can be labeled inside or outside $S(Q)$ using a full violation test.  This implies, following Algorithm \ref{alg:count-prob}, that we can enumerate all $O((nk)^\beta)$ potential bases, and for each determine its weight towards the SIP function using the full violation test.  This procedure generates a set of $O((nk)^\beta)$ weighted shapes in $O((nk)^{\beta+1})$ time.  Finally, a query to the SIP function can be evaluated by counting the weighted fraction of shapes that it is contained in.  

We can also build a data structure to speed up the query time.  Since the boundary of each pair of shapes intersects a constant number of times, then all $O((nk)^{2\beta})$ pairs intersect at most $O((nk)^{2\beta})$ times in total.  In $\b{R}^2$, the arrangement of these shapes forms a planar subdivision with as many regions as intersection points.  We can precompute the weighted fraction of shapes overlapping on each region.  Textbook techniques can be used to build a query structure of size $O((nk)^{2\beta})$ that allows for stabbing queries in time $O(\log(nk))$ to determine which region the query lies, and hence what the associated weighted fraction of points is, and what the SIP value is.  
We summarize these results in the following theorem.  

\begin{theorem}
Consider a set $\Eu P$ of $n$ indecisive point sets of size $k$ each, and an LP-type problem $f_S : \Eu P \to \R$ with combinatorial dimension $\beta$ that finds a summarizing shape for which every pair intersects a constant number of times.  
We can create a data structure of size $O((nk)^\beta)$ in time $O((nk)^{\beta+1})$ that answers SIP queries exactly in $O((nk)^\beta)$ time, 
and another structure of size $O((nk)^{2\beta})$ in time $O((nk)^{2\beta})$ that answers SIP queries exactly in $O(\log(nk))$ time.  
\end{theorem}

Furthermore, through specific invocations of Theorem \ref{thm:n-apx}, we can extend these polynomial deterministic approaches to create $\eps$-SIP data structures for many natural classes of uncertain input point sets.

\section {Conclusions}

In this paper, we studied the computation and representation of complete probability distributions on the output of single- and multi-valued geometric functions, when the input points are uncertain. 
We considered randomized and deterministic, exact and approximate approaches for indecisive and probabilistic uncertain points, and presented polynomial-time algorithms as well as hardness results.
These results extend to when the output distribution is over the family of low-description-complexity summarizing shapes. 

We draw two main conclusions.
Firstly, we observe that the tractability of exact computations on indecisive points really depends on the problem at hand. On the one hand, the output distribution of LP-type problems can be represented concisely and computed efficiently. On the other hand, even computing a single value of the output distribution of the diameter problem is already \#P-hard.
Secondly, we showed that computing approximate quantizations deterministically is often possible in polynomial time. However, the polynomials in question are of rather high degree, and while it is conceivable that these degrees can be reduced further, this will require some new ideas. In the mean time, the randomized alternatives remain more practical.

We believe that the problem of representing and approximating distributions of more complicated objects (especially when they are somehow representing uncertain points), is an important direction for further study.

\section*{Acknowledgments}
The authors would like to thank Joachim Gudmundsson and Pankaj Agarwal for helpful discussions in early phases of this work, Sariel Har-Peled for discussions about wedges, and Suresh Venkatasubramanian for organizational tips.  

This research was partially supported by the Netherlands Organisation for
Scientific Research (NWO) through the project GOGO and under grant 639.021.123,
by the Office of Naval Research under MURI grant N00014-08-1-1015,
and by a subaward to the University of Utah under NSF award 0937060 to CRA.

\newpage
\bibliographystyle{abuser}
\bibliography{uncert,refs}

\newpage
\appendix

\section{$\eps$-Samples of Distributions}
\label{sec:eps-dist}
In this section we explore conditions for continuous distributions such that they can be approximated with bounded error by discrete distributions (point sets).  We state specific results for multi-variate normal distributions.  

We say a subset $W \subset \R^d$ is \emph{polygonal approximable} if there exists a polygonal shape $S \subset \R^d$ with $m$ facets such that $\phi(W \setminus S) + \phi(S \setminus W) \leq \eps \phi(W)$ for any $\eps > 0$.  Usually, $m$ is dependent on $\eps$, for instance for a d-variate normal distribution $m = O((1/\eps^{d+1}) \log (1/\eps))$~\cite{Phi08,Phi09}.  In turn, such a polygonal shape $S$ describes a continuous point set where $(S, \Eu{A})$ can be given an $\eps$-sample $Q$ using $O((1/\eps^2) \log (1/\eps))$ points if $(S, \Eu{A})$ has bounded VC-dimension \cite{Mat99} or using $O((1/\eps) \log^{2k} (1/\eps))$ points if $\Eu{A}$ is defined by a constant $k$ number of directions \cite{Phi08}.  For instance, where $\Eu{A} = \Eu{B}$ is the set of all balls then the first case applies, and when $\Eu{A} = \Eu{R}_2$ is the set of all axis-aligned rectangles then either case applies.

A shape $W \subset \R^{d+1}$ may describe a distribution $\mu : \R^d \to [0,1]$.
For instance for a range space $(\mu, \Eu{B})$, then the range space of the associated shape $W_\mu$ is $(W_\mu, \Eu{B} \times \R)$ where $\Eu{B} \times \R$ describes balls in $\R^d$ for the first $d$ coordinates and any points in the $(d+1)$th coordinate.

The general scheme to create an $\eps$-sample for $(S, \Eu{A})$, where $S \in \R^d$ is a polygonal shape, is to use a lattice $\Lambda$ of points.
A \emph{lattice} $\Lambda$ in $\R^d$ is an infinite set of points defined such that for $d$ vectors $\{v_1, \ldots, v_d\}$ that form a basis, for any point $p \in \Lambda$, $p + v_i$ and $p - v_i$ are also in $\Lambda$ for any $i \in [1,d]$.
We first create a discrete $(\eps/2)$-sample $M = \Lambda \cap S$ of  $(S, \Eu{A})$ and then create an $(\eps/2)$-sample $Q$ of $(M, \Eu{A})$ using standard techniques~\cite{CM96,Phi08}.  Then $Q$ is an $\eps$-sample of $(S, \Eu{A})$.
For a shape $S$ with $m$ $(d-1)$-faces on its boundary, any subset $A^\prime \subset \R^d$ that is described by a subset from $(S, \Eu{A})$ is an intersection $A^\prime = A \cap S$ for some $A \in \Eu{A}$.  Since $S$ has $m$ $(d-1)$-dimensional faces, we can bound the VC-dimension of $(S, \Eu{A})$ as $\nu = O((m + \nu_{\Eu{A}}) \log (m+\nu_{\Eu{A}}))$ where $\nu_{\Eu{A}}$ is the VC-dimension of $(\R^d, \Eu{A})$.
Finally the set $M = S \cap \Lambda$ is determined by choosing an arbitrary initial origin point in $\Lambda$ and then uniformly scaling all vectors $\{v_1, \ldots, v_d\}$ until $|M| = \Theta((\nu/\eps^2) \log (\nu/\eps))$~\cite{Mat99}.  This construction follows a less general but smaller construction in Phillips~\cite{Phi08}.

It follows that we can create such an $\eps$-sample $K$ of $(S, \Eu A)$ of size $|M|$ in time $O(|M| m \log |M|)$ by starting with a scaling of the lattice so a constant number of points are in $S$ and then doubling the scale until we get to within a factor of $d$ of $|M|$.  If there are $n$ points inside $S$, it takes $O(nm)$ time to count them.  We can then take another $\eps$-sample of $(K, \Eu A)$ of size
$O((\nu_{\Eu A}/\eps^2) \log (\nu_{\Eu A} /\eps))$
 in time
 $O(\nu_{\Eu A}^{3 \nu_{\Eu A}} |M| ((1/\eps^2) \log (\nu_{\Eu A} / \eps))^{\nu_{\Eu A}})$.

\begin{theorem}
For a polygonal shape $S \subset \R^d$ with $m$ (constant size) facets, we can construct an $\eps$-sample for $(S, \Eu{A})$ of size
$O((\nu/\eps^2) \log (\nu/\eps))$
in time
$O(m (\nu/\eps^2) \log^2 (\nu/\eps))$,
where $(S, \Eu{A})$ has VC-dimension $\nu_{\Eu A}$ and
$\nu = O((\nu_{\Eu{A}}+m) \log (\nu_{\Eu{A}}+m))$.

This can be reduced to size
$O((\nu_{\Eu A}/\eps^2) \log (\nu_{\Eu A}/\eps))$
in time
$O((1/\eps^{2\nu_{\Eu A} + 2}) (m+\nu_{\Eu A}) \log^{\nu_{\Eu A}}((m + \nu_{\Eu A})/\eps) \log (m/\eps)$.
\label{lem:mu-eps}
\end{theorem}

We can consider the specific case of when $W \subset \R^3$ is a $d$-variate normal distribution $\mu : \R^d \to \R^+$.  Then $m = O((1/\eps^d) \log (1/\eps))$ and $|M| = O((m/\eps^2) \log m \log (m/\eps)) = O((1/\eps^{d+2}) \log^3 (1/\eps))$.

\begin{corollary}
Let $\mu : \R^d \to \R^+$ be a $d$-variate normal distribution with constant standard deviation.
We can construct an $\eps$-sample of $(\mu, \Eu A)$ with VC-dimension $\nu_{\Eu A} \geq 2$ (and where $d \leq \nu_{\Eu A} \leq 1/\eps$) of size $O((\nu_{\Eu A} / \eps^2) \log (\nu_{\Eu A}/\eps))$ in time
$O((1/\eps^{2\nu_{\Eu A} + d + 2}) \log^{\nu_{\Eu A}+1} (1/\eps))$.
\label{cor:norm-es}
\end{corollary}

For convenience, we restate a tighter, but less general theorem from Phillips, here slightly generalized.
\begin{theorem}[\cite{Phi08,Phi09}]
Let $\mu : \R^d \to \R^+$ be a $d$-variate normal distribution with constant standard deviation.
Let $(\mu, \Eu Q_k)$ be a range space where the ranges are defined as the intersection of $k$ slabs with fixed normal directions.
We can construct an $\eps$-sample of $(\mu, \Eu Q_k)$ of size $O((1/\eps)\log^{2k} (1/\eps))$ in time $O((1/\eps^{d+4}) \log^{6k + 3} (1/\eps))$.
\label{thm:Phi08}
\end{theorem}

\subsection{Avoiding Degeneracy}
An important part of the above construction is the arbitrary choice of the origin points of the lattice $\Lambda$.  This allows us to arbitrarily shift the lattice defining $M$ and thus the set $Q$.  In Section \ref{sec:pta} we need to construct $n$ $\eps$-samples $\{Q_1, \ldots, Q_n\}$ for $n$ range spaces $\{(S_1, \Eu{A}), \ldots, (S_n, \Eu{A})\}$.
In Algorithm \ref{alg:count-prob} we examine sets of $\nu_{\Eu{A}}$ points, each from separate $\eps$-samples that define a minimal shape $A \in \Eu{A}$.  It is important that we do not have two such (possibly not disjoint) sets of $\nu_{\Eu{A}}$ points that define the same minimal shape $A \in \Eu{A}$.  (Note, this does not include cases where say two points are antipodal on a disk and any other point in the disk added to a set of $\nu_{\Eu{A}}=3$ points forms such a set; it refers to cases where say four points lie (degenerately) on the boundary of a disc.)
We can guarantee this by enforcing a property on all pairs of origin points $p$ and $q$ for $(S_i, \Eu{A})$ and $(S_j, \Eu{A})$.  For the purpose of construction, it is easiest to consider only the $l$th coordinates $p_l$ and $q_l$ for any pair of origin points or lattice vectors (where the same lattice vectors are used for each lattice).  We enforce a specific property on every such pair $p_l$ and $q_l$, for all $l$ and all distributions and lattice vectors.

First, consider the case where $\Eu{A} = \Eu{R}_d$ describes axis-aligned bounding boxes.  It is easy to see that if for all pairs $p_l$ and $q_l$ that $(p_l - q_l)$ is irrational, then we cannot have $>2d$ points on the boundary of an axis-aligned bounding box, hence the desired property is satisfied.

Now consider the more complicated case where $\Eu{A} = \Eu{B}$ describes smallest enclosing balls.  There is a polynomial of degree $2$ that describes the boundary of the ball, so we can enforce that for all pairs $p_l$ and $q_l$ that $(p_l - q_l)$ is of the form $c_1 (r_{p_l})^{1/3} + c_2 (r_{q_l})^{1/3}$ where $c_1$ and $c_2$ are rational coefficients and $r_{p_l}$ and $r_{q_l}$ are distinct integers that are not multiple of cubes.    Now if $\nu = d+1$ such points satisfy (and in fact define) the equation of the boundary of a ball, then no $(d+2)$th point which has this property with respect to the first $d+1$ can also satisfy this equation.

More generally, if $\Eu{A}$ can be described with a polynomial of degree $p$ with $\nu$ variables, then enforce that every pair of coordinates are the sum of $(p+1)$-roots.  This ensures that no $\nu+1$ points can satisfy the equation, and the undesired situation cannot occur.

\section {Computing Other Measures on Uncertain Points}
\label {app:moreexamples}
We generalize this machinery to other LP-type problems $f$ defined on a set of points in $\R^d$ and with constant combinatorial dimension.  Although, in some cases (like smallest axis-aligned bounding box by perimeter) we are able to show that $(\R^d, \Eu A_{f,n})$ has constant VC-dimension, for other cases (like radius of the smallest enclosing disk) we cannot and need to first decompose each range $A(Z, w) \in \Eu A_{f,n}$ into a set of disjoint ``wedges'' from a family of ranges $\Eu W_f$.

To simplify the already large polynomial runtimes below we replace the runtime bound in Theorem \ref{thm:gen-comb-dist-bignum} (below, which extends Theorem \ref{thm:gen-comb-dist} to deal with large input sizes) with $O((nk)^{\beta+1} \log^4(nk))$.

\begin{lemma}
If the disjoint union of $m$ shapes from $\Eu{W}_{f}$ can form any shape from $\Eu{A}_{f,n}$, then an $(\eps/m)$-sample of $(M, \Eu{W}_{f})$ is an $\eps$-sample of $(M, \Eu{A}_{f,n})$.
\end{lemma}
\begin{proof}
For any shape $A \in \Eu{A}_{f,n}$ we can create a set of $m$ shapes $\{W_1, \ldots, W_n\} \subset \Eu{W}_{f}$ whose disjoint union is $A$.  Since each range of $\Eu{W}_{f}$ may have error $\eps/m$, their union has error at most $\eps$.
\end{proof}

We study several example cases for which we can deterministically compute $\eps$-quantizations.  For each case we show an example element of $\Eu{A}_{f,n}$ on an example of $7$ points.

To facilitate the analysis, we define the notion of shatter dimension, which is similar to VC-dimension.  The shatter function $\pi_T(m)$ of a range space $T = (Y,\Eu A)$ is the maximum number of sets in $T$ where $|Y| = m$.  The \emph{shatter dimension} $\sigma_T$ of a range space $T = (Y,\Eu A)$ is the minimum value such that $\pi_T(m) = O(m^{\sigma_T})$.  If a range $A \in \Eu A$ is defined by $k$ points, then $k \geq \sigma_T$.  It can be shown \cite{HP} that $\sigma_T \leq \nu_T$ and $\nu_T = O(\sigma_T \log \sigma_T)$.
And, in general, the basis size of the related LP-type problem is bounded $\beta \leq \sigma_T$.

\paragraph{Directional Width.}
We first consider the problem of finding the width along a particular direction $u$ (\textsf{dwid}).  Given a point set $P$, $f(P)$ is the width of the minimum slab containing $P$, as in Figure \ref{fig:shape-dwid}.    This can be thought of as a one-dimensional problem by projecting all points $P$ using the operation $\IP{\cdot}{u}$.  The directional width is then just the difference between the largest point and the smallest point.  As such, the VC-dimension of $(\R^d, \Eu A_f)$ is $2$.  Furthermore, $\Eu{A}_{f,n} = \Eu A_f$ in this case, so $(\R^d, \Eu{A}_{f,n})$ also has VC-dimension $2$.  For $1 \leq i \leq n$, we can then create an $(\eps/n)$-sample $Q_i$ of $(\mu_i, \Eu{A}_{f,n})$ of size $k = O(n/\eps)$ in $O((n/\eps) \log (n/\eps))$ time given basic knowledge of the distribution $\mu_i$.
We can then apply Theorem \ref{thm:gen-comb-dist} to build an $\eps$-quantization in time $O((kn)^{2+1}\log^4(n^2/\eps)) = O(n^6/\eps^3 \log^4 (n/\eps))$ for the \textsf{dwid} case.

We can actually evaluate $\tilde f (\Eu Q, r)$ for all values of $r$ faster using a series of sweep lines.  Each of the $O(n^4/\eps^2)$ potential bases are defined by a left and right end point.  Each of the $O(n^2/\eps)$ points in $\bigcup_i Q_i$ could be a left or right end point.
We only need to find the $O(1/\eps)$ widths (defined by pairs of end points) that wind up in the final $\eps$-quantization.
Using a Frederickson and Johnson approach~\cite{FJ84}, we can search for each width in $O(\log (n/\eps))$ steps.  At each step we are given a width $\omega$ and need to decide what fraction of \psets have width at most $\omega$.  We can scan from each of the possible left end points and count the number of \psets that have width at most $\omega$.  For each $\omega$, this can be performed in $O(n^2/\eps)$ time with a pair of simultaneous sweep lines.
The total runtime is $O(1/\eps) \cdot O(\log (n/\eps)) \cdot O(n^2/\eps) = O((n^2/\eps^2) \log (n/\eps))$.

\begin{theorem}
We can create an $\eps$-quantization of size $O(1/\eps)$ for the \textsf{dwid} problem in $O((n^2/\eps^2) \log (n/\eps))$ time.
\end{theorem}

\tweeplaatjes {shape-dwid} {shape-aabba}
{(a) Directional (vertical) width. (b) Axis-aligned bounding box, measured by area. The curves are hyperbola parts.}

\paragraph{Axis-aligned bounding box.}
We now consider the set of problems related to axis-aligned bounding boxes in $\R^d$.  For a point set $P$, we minimize $f(P)$, which either represents the $d$-dimensional volume of $S(P)$ (the \textsf{aabbv} case --- minimizes the area in $\R^2$) or the $(d-1)$-dimensional volume of the boundary of $S(P)$ (the \textsf{aabbp} case --- minimizes the perimeter in $\R^2$).
      Figures~\ref{fig:shape-aabbp} and~\ref {fig:shape-aabba} show two examples of elements of $\Eu{A}_{f,n}$ for the \textsf{aabbp} case and the \textsf{aabbv} case in $\R^2$.  For both $(\R^2, \Eu{A}_{f,n})$ has a shatter dimension of $4$ because the shape is determined by the $x$-coordinates of $2$ points and the $y$-coordinates of $2$ points.  This generalizes to a shatter dimension of $2d$ for $(\R^d, \Eu{A}_{f,n})$, and hence a VC-dimension of $O(d \log d)$.  The smaller VC-dimension in the \textsf{aabbp} case discussed in detail above can be extended to higher dimensions.

Hence, for $1 \leq i \leq n$,  for both cases we can create an $(\eps/n)$-sample $Q_i$ of $(\mu_i, \Eu{A}_{f,n})$, each of size $k = O((n^2/\eps^2) \log (n/\eps))$ in total time $O(((n/\eps) \log (n/\eps))^{O(d \log d)})$ via Corollary \ref{cor:norm-es}.
In $\R^d$, we can construct the $\eps$-quantization in
$O((kn)^{2d+1} \log^4 (nk)) = O((n^{6d+3}/\eps^{4d+2}) \log^{2d+1} (n/\eps))$ time via Theorem \ref{thm:gen-comb-dist}.

\begin{theorem}
We can create an $\eps$-quantization of size $O(1/\eps)$ for the \textsf{aabbp} or \textsf{aabbv} problem on $n$ $d$-variate normal distributions in $O(((n/\eps) \log (n/\eps))^{O(d \log d)})$ time.
\end{theorem}

\paragraph{Smallest enclosing ball.}
Figure \ref{fig:shape-seblinfty+shape-sebl1} shows example elements of $\Eu{A}_{f,n}$ for smallest enclosing ball, for metrics $L_\infty$ (the \textsf{seb}$_{\infty}$ case) and $L_1$ (the \textsf{seb}$_1$ case) in $\R^2$.  An example element of $\Eu{A}_{f,n}$ for smallest enclosing ball for the $L_2$ metric (the \textsf{seb}$_2$ case) was shown in Figure \ref{fig:shape-sebl2}.  For $\textsf{seb}_\infty$ and $\textsf{seb}_1$, $(\R^d, \Eu{A}_{f,n})$ has VC-dimension $2d$ because the shapes are defined by the intersection of halfspaces from $d$ predefined normal directions.
For $\textsf{seb}_1$ and $\textsf{seb}_\infty$, we can create $n$ $(\eps/n)$-samples $Q_i$ of each $(\mu_i, \Eu{A}_{f,n})$ of size $k = O((n/\eps) \log^{2d} (n/\eps))$ in total time $O(n (n/\eps)^{d+4} \log^{6k + 3} (n/\eps))$ via Theorem \ref{thm:Phi08}.
We can then create an $\eps$-quantization in $O((nk)^{2d+1} \log^4 (nk)) = O((n^{4d+2}/\eps^{2d+1} )\log^{4d^2 + 2d + 4} (n/\eps))$ time via Theorem \ref{thm:gen-comb-dist}.

\tweeplaatjes {shape-seblinfty} {shape-sebl1}
{(a) Smallest enclosing ball, $L_\infty$ metric. (b) Smallest enclosing ball, $L_1$ metric.}

\begin{theorem}
We can create an $\eps$-quantization of size $O(1/\eps)$ for the \textsf{seb$_1$} or \textsf{seb$_\infty$} problem on $n$ $d$-variate normal distributions in $O((n^{4d+2}/\eps^{2d+1}) \log^{4d^2 + 2d + 4} (n/\eps))$ time.
\end{theorem}

For the $\textsf{seb}_2$ case in $\R^2$,
$(\R^2, \Eu{A}_{f,n})$ has infinite VC-dimension, but $(\R^2, \Eu{W}_{f})$ has VC-dimension at most $9$ because it is the intersection of $2$ halfspaces and one disc.
Any shape $A(T,w) \in \Eu{A}_{f,n}$ can be formed from the disjoint union of $2n$ wedges.  Choosing a point in the convex hull of $T$ as the vertex of the wedges will ensure that each wedge is completely inside the ball that defines part of its boundary.  Thus, in $\R^2$ the $n$ $(\eps/n)$-samples of each $(\mu_{p_i}, \Eu{A}_{f,n})$ are of size $\lambda_f(n,\eps) = O((n^4/\eps^2) \log (n/\eps))$ and can all be calculated in total time $O((n^{5}/\eps^{2}) \log^2 (n/\eps))$.  And then
the $\eps$-quantization can be calculated in $O(n^{56/3}/\eps^{22/3} \log^{11/3} (n/\eps))$ time by assuming general position of all $Q_i$ and then using range counting data structures.
We conjecture this technique can be extended to $\R^d$, but we cannot figure how to decompose a shape $A \in \Eu{A}_{f,n}$ into a polynomial number of wedges with constant VC-dimension.

       \begin{table}[h!!t]
      \caption{$\eps$-Samples for Summarizing Shape Family $\Eu{A}_{f,n}$.}
      \small
      \centering
      \begin{tabular}{|l|c|c|c|c|c|}
      \hline
      case & $k$ & $\nu_T$ & $(n k)^{\nu_T}$ & $\textsf{RC}(k,\Eu A_f)$ & runtime
      \\ \hline \hline
      \textsf{dwid} & $O(n/\eps)$ & $2$ & $O(n^4/\eps^2)$ & $\tilde{O}(1)$ & $\tilde{O}(n^5/\eps^2)$
      \\ \hline
      \textsf{aabbp} & $\tilde{O}(n^2/\eps^2)$ & $2d$ & $\tilde{O}(n^{6d}/\eps^{4d})$ & $\tilde{O}(1)$ & $\tilde{O}(n^{6d+1}/\eps^{4d})$
      \\ \hline
      \textsf{aabbv} & $\tilde{O}(n^2/\eps^2)$ & $2d$ & $\tilde{O}(n^{6d}/\eps^{4d})$ & $\tilde{O}(1)$ & $\tilde{O}(n^{6d+1}/\eps^{4d})$
      \\ \hline
      \textsf{seb$_\infty$} & $\tilde{O}(n/\eps)$ & $d+1$ & $\tilde{O}(n^{2d+2}/\eps^{d+1})$ & $\tilde{O}(1)$ & $\tilde{O}(n^{2d+3}/\eps^{d+1})$
      \\ \hline
      \textsf{seb$_1$} & $\tilde{O}(n/\eps)$ & $d+1$ & $\tilde{O}(n^{2d+2}/\eps^{d+1})$ & $\tilde{O}(1)$ & $\tilde{O}(n^{2d+3}/\eps^{d+1})$
      \\ \hline
      \textsf{seb$_2$} $\in \R^2$ & $\tilde{O}(n^4/\eps^2)$ & $3$ & \hspace{-.06in}$\tilde{O}(n^{15} / \eps^6)$ & $\tilde{O}(n^{8/3}/\eps^{4/3})$ & $\tilde{O}(n^{56/3}/\eps^{22/3}))$
      \\ \hline
      \textsf{diam} $\in \R^2$ & $\tilde{O}(n^4/\eps^2)$  & $n$ & $\tilde{O}(n^5/\eps^2)^n$ & $\tilde{O}(n^4 / \eps^2)$ & $\tilde{O}((n^5/\eps^2)^{n+1})$
      \\ \hline
     \end{tabular}
      \label{tbl:Afn-size}
      \\  $\tilde{O}(f(n,\eps))$ ignores poly-logarithmic factors $(\log (n/\eps))^{O(\textrm{poly}(d))}$.
      \end{table}

\section{Algorithm \ref{alg:count-prob} in RAM model}
\label{sec:bignum}
In this section we will analyze Algorithm \ref{alg:count-prob} without the assumption that the integer $k^n$ can be stored in $O(1)$ words.  
To simplify the results, we assume a RAM model where a word size contains $O(\log (nk))$ bits, each weight $w(q_{i,j})$ can be stored in one word, and the weight of any basis $f(B)$, the product of $n$ $O(1)$ word weights, can thus be stored in $O(n)$ words.
The following lemma describes the main result we will need relating to large numbers.
\begin{lemma}
In the RAM model where a word size has $b$ bits we can calculate the product of $n$ numbers where each is described by $O(b)$ bits in $(n b \log^2 n) 2^{O(\log^* n)}$ time.
\end{lemma}
\begin{proof}
Using F\"urer's recent result~\cite{Fur09} we can multiply two $m$-bit numbers in $m \log m 2^{O(\log^* m)}$ bit operations.  The product of $n$ $m$-bit numbers has $O(\log ((2^m)^n) = O(nm)$ bits, and can be accomplished with $n-1$ pairwise multiplications.

We can calculate this product efficiently from the bottom up, starting with $n/2$ multiplications of two $m=O(b)$ bit numbers.  Then we perform $n/4$ multiplications of two $2m$ bit numbers, and so on.  Since each operation on $O(b)$ bits takes $O(1)$ time in our model, F\"urer's result clearly upper bounds the RAM result.
The total cost of this can be written as
\[
\sum_{i=0}^{\log n} \frac{n}{2^{i-1}} (2^i m) \log (2^i m) 2^{O(\log^* (2^i m))}
=
2nb \sum_{i=0}^{\log n} (i+\log b) 2^{O(\log^* (2^i\log b))}
=
nb \log n \log (nb) 2^{O(\log^* (nb))}.
\]
Since we assume $b = O(\log n)$ we can simplify this bound to $(nb \log^2 n) 2^{O(\log^*n)}$.
\end{proof}

This implies that Algorithm \ref{alg:count-prob} takes $O((nk)^{\beta+1}) + ((nk)^{\beta}n \log (nk) \log^2 n)2^{O(\log^* n)}$ time where each $w_i$ can be described in $O(b) = O(\log (nk))$ bits.  This dominates the single division and all other operations.  
Now we can rewrite Theorem \ref{thm:gen-comb-dist} without the restriction that $k^n$ can be stored in $O(1)$ words.  

\begin{theorem}
Given a set $\Eu Q$ of $n$ indecisive point sets of size $k$ each, and given an LP-type problem $f : \Eu Q \to \R$ with combinatorial dimension $\beta$, we can create the distribution of $f$ over $\Eu Q$ in $O((nk)^{\beta+1}) + ((nk)^{\beta} n \log(nk) \log^2 n)2^{O(\log^* n)}$ time.  The size of the distribution is $O(n (nk)^\beta)$.
\label{thm:gen-comb-dist-bignum}
\end{theorem}

\end{document}